\documentclass[11pt]{article}
\usepackage{etex,etoolbox}
\usepackage{fullpage}
\usepackage[utf8]{inputenc}
\usepackage[T1]{fontenc}
\usepackage[svgnames]{xcolor}
\usepackage{graphicx}
\usepackage{amsmath, amssymb, amstext}
\usepackage{amsthm}
\usepackage{tikz}
\usepackage{mathtools}
\usepackage{stmaryrd}

\usepackage{pgfplots}
\usepackage{float}

\usepackage[colorlinks]{hyperref}
\definecolor{lightblue}{rgb}{0.5,0.5,1.0}
\definecolor{darkred}{rgb}{0.5,0,0}
\definecolor{darkgreen}{rgb}{0,0.5,0}
\definecolor{darkblue}{rgb}{0,0,0.5}

\hypersetup{colorlinks,linkcolor=darkred,filecolor=darkgreen,urlcolor=darkred,citecolor=darkblue}

\usepackage[titlenumbered, linesnumbered, ruled]{algorithm2e}

\usetikzlibrary{patterns,arrows,decorations.pathreplacing,decorations.pathmorphing,calc}

\tikzstyle{normalvertex}=[circle,fill=White,draw=Black]

\allowdisplaybreaks

\theoremstyle{definition}
\newtheorem{theorem}{Theorem}[section]
\newtheorem{lemma}[theorem]{Lemma}
\newtheorem{definition}[theorem]{Definition}
\newtheorem{corollary}[theorem]{Corollary}

\newtheorem{proposition}[theorem]{Proposition}

\DeclareMathOperator{\Aut}{Aut}

\DeclareMathOperator{\rk}{rk}

\DeclareMathOperator{\Sol}{Sol}

\DeclareMathOperator{\cl}{cl}
\DeclareMathOperator{\attr}{attr}

\DeclareMathOperator{\BP}{BP}

\DeclareMathOperator{\WL}{\texttt{WL}}
\DeclareMathOperator{\sel}{\texttt{sel}}
\DeclareMathOperator{\refine}{\texttt{ref}}
\DeclareMathOperator{\inv}{\texttt{inv}}

\numberwithin{equation}{section}
\numberwithin{figure}{section}

\title{An exponential lower bound for Individualization-Refinement algorithms for Graph Isomorphism}
\author{Daniel Neuen and Pascal Schweitzer\\ 
RWTH Aachen University\\
\texttt{\{neuen,schweitzer\}@informatik.rwth-aachen.de}
}

\newcounter{claimcounter}
\newenvironment{claim}[1][]{
  \renewcommand{\proof}{\smallskip\par\noindent\textit{Proof. }}
  \medskip\par\noindent%
  \ifthenelse{\equal{#1}{}}{%
    \setcounter{claimcounter}{0}\refstepcounter{claimcounter}\textit{Claim~\arabic{claimcounter}.}
  }{%
    \ifthenelse{\equal{#1}{resume}}{%
      \refstepcounter{claimcounter}\textit{Claim~\arabic{claimcounter}.}
    }{%
      \textit{Claim~#1.}
    }
  }
}{
  \par\medskip
}
\newcommand{\uend}{\hfill$\lrcorner$}

\newcommand{\notleftright}{\mathrel{\ooalign{$\Leftrightarrow$\cr\hidewidth$/$\hidewidth}}}

\begin{document}

\maketitle

\begin{abstract}
 The individualization-refinement paradigm provides a strong toolbox for testing isomorphism of two graphs and indeed, the currently fastest implementations of isomorphism solvers all follow this approach. 
 While these solvers are fast in practice, from a theoretical point of view, no general lower bounds concerning the worst case complexity of these tools are known. In fact, it is an open question whether individualization-refinement algorithms can achieve upper bounds on the running time similar to the more theoretical techniques based on a group theoretic approach.
 
 In this work we give a negative answer to this question and construct a family of graphs on which algorithms based on the individualization-refinement paradigm require exponential time.
 Contrary to a previous construction of Miyazaki, that only applies to a specific implementation within the individualization-refinement framework, our construction is immune to changing the cell selector, or adding various heuristic invariants to the algorithm.
 Furthermore, our graphs also provide exponential lower bounds in the case when the $k$-dimensional Weisfeiler-Leman algorithm is used to replace the standard color refinement operator and the arguments even work when the entire automorphism group of the inputs is initially provided to the algorithm.
\end{abstract}

\section{Introduction}

The individualization-refinement paradigm provides a strong toolbox for testing isomorphism of two graphs. To date, algorithms that implement the individualization-refinement paradigm constitute the fastest practical algorithms for the graph isomorphism problem and for the task of canonically labeling combinatorial objects.

Originally exploited by McKay's software package nauty~\cite{MR635936} as early as 1981, in a nutshell, the basic principle is to classify vertices using a refinement operator according to an isomorphism-invariant property. In a basic form one usually uses the so-called color refinement operator, also called 1-dimensional Weisfeiler-Leman algorithm, for this purpose. Whenever the refinement is not sufficient, vertices within a selected color class (usually called a cell) are individualized one by one in a backtracking manner as to artificially distinguish them from other vertices. This yields a backtracking tree, that is traversed to explore the structure of the input graphs.
Additional pruning with the use of invariants and the exploitation of automorphisms of the graphs makes the approach viable in practice, leading to the fastest isomorphism solvers currently available. The use of invariants also allows us to define a smallest leaf, which can be used to canonically label the graph, i.e., to rearrange the vertices in canonical fashion as to obtain a standard copy of the graph.

There are several highly efficient isomorphism software packages implementing the paradigm. Among them are nauty/traces~\cite{mckay}, bliss~\cite{bliss}, conauto~\cite{conauto} and saucy~\cite{saucy}. While they all follow the basic individualization-refinement paradigm, these algorithms differ drastically in 
design principles and algorithmic realization. In particular, they differ in the way the search tree is traversed, they use different low level subroutines, have diverse ways to perform tasks such as automorphism detection, and they use different cell selection strategies as well as vertex invariants and refinement operators. 

With Babai's \cite{DBLP:conf/stoc/Babai16} recent quasi-polynomial time algorithm for the graph isomorphism problem, the theoretical worst case complexity of algorithms for the graph isomorphism problem was drastically improved from a previous best~$e^{O(\sqrt{n\log n })}$ (see \cite{DBLP:conf/stoc/BabaiL83}) to~$O(n^{\log^c n})$ for some constant~$c\in \mathbb{N}$.
As an open question, Babai asks~\cite{DBLP:conf/stoc/Babai16} for the worst case complexity of algorithms based on individualization-refinement techniques. About this worst case complexity, very little had been known.

In 1995 Miyazaki~\cite{DBLP:conf/dimacs/Miyazaki95} constructed a family of graphs on which the then current implementation of nauty has exponential running time.
For this purpose these graphs are designed to specifically fool the cell selection process into exponential behavior.
However, as Miyazaki also argues, with a different cell selection strategy the examples can be solved in polynomial time within the individualization-refinement paradigm.

In this paper we provide general lower bounds for individualization-refinement algorithms with arbitrary combinations of cell selection, refinement operators, invariants and even given perfect automorphism pruning. 
More precisely, the graphs we provide yield an exponential 
size search tree (i.e., $2^{\Omega(n)}$ nodes) for any combination of refinement operator, invariants, and the cell selector which are not stronger than the~$k$-dimensional Weisfeiler-Leman algorithm for some fixed dimension~$k$. The natural class of algorithms for which we thus obtain lower bounds encompasses all software packages mentioned above even with various combinations of switches that can be turned on and off in the execution of the algorithm to tune the algorithms towards specific input graphs. 
Our graphs are asymmetric, i.e., have no non-trivial automorphisms, and thus no strategy for automorphism detection can help the algorithm to circumvent the exponential lower bound.

Our construction makes use of a construction of Cai-Fürer-Immerman \cite{cfi} and the multipede construction of Gurevich and Shelah \cite{DBLP:journals/jsyml/GurevichS96} that yields for every dimension~$k$ non-isomorphic finite rigid structures that are not distinguishable by the~$k$-dimensional Weisfeiler-Leman algorithm.  
In more detail, our construction starts with a bipartite base graph that is obtained by a simple random process.  With high probability such a graph has strong expansion properties ensuring a variant of the meagerness property of~\cite{DBLP:journals/jsyml/GurevichS96} suitable for our purposes. Additionally, with high probability the graph has an almost-disjointness property for neighborhoods of vertices from one bipartition class.
To the base graph we apply a bipartite variant of the construction of~\cite{cfi}.
By individualizing a small fraction of vertices, we can guarantee that the final graphs are rigid (have no non-trivial automorphisms).  
For our theoretical analysis, we define a closure operator that gives us control over the effect of the Weifeiler-Leman algorithm on the graphs. Due to the disjointness property this effect is limited.
Exploiting automorphism of subgraphs of the input, we then proceed to argue that there is an exponential number of colorings of the graph that cannot be distinguished. 
These statements can be combined to show that the search tree of every algorithm within the individualization-refinement framework has exponential size.

Some of the packages above have a mechanism called component recursion (see~\cite{DBLP:conf/tapas/JunttilaK11}). 
We show that even this strategy cannot yield improvements for our examples.
We should point out that component recursion was used in Goldberg's result~\cite{Goldberg1983229} which shows that with the right cell selection strategy and the use of component recursion (in that paper called sections) individualization-refinement algorithms have exponential upper bounds, matching our lower bounds.

We also should remark that, seen as colored graphs, our graphs have bounded color class size and as such isomorphism of the graphs can be decided in polynomial time using simple group theoretic techniques (see \cite{BabaiRandom, DBLP:conf/focs/FurstHL80}).

Since the software packages that follow the individualization-refinement paradigm are designed for practical purposes rather than to obtain theoretical worst case guarantees, the question lies at hand how meaningful the lower bounds provided in the paper are. However, in separate work \cite{benchmark-paper}, we investigate practical benchmark graphs. It turns out that constructions related to the ones discussed in this paper in fact yield graphs which, experimentally, pose by far the most challenging graph isomorphism instances available to date.

\section{Preliminaries}

\subsection{Graphs}

A \emph{graph} is a pair $G=(V,E)$ with vertex set $V = V(G)$ and edge relation $E = E(G)$.
In this paper all graphs are finite simple, undirected graphs. The \emph{neighborhood} of~$v\in V(G)$ is denoted~$N(v)$.
For a set $X \subseteq V$ let $N(X) = (\bigcup_{v \in X} N(v)) \setminus X$.

An \emph{isomorphism} from a graph $G$ to another graph $H$ is a bijective mapping $\varphi\colon V(G) \rightarrow V(H)$ which preserves the edge relation, that is~$\{v,w\} \in E(G)$ if and only if $\{\varphi(v),\varphi(w)\} \in E(H)$ for all~$v,w \in V(G)$.
Two graphs $G$ and $H$ are \emph{isomorphic} ($G \cong H$) if there is an isomorphism from~$G$ to~$H$.
We write~$G \stackrel{\varphi}{\cong} H$ to indicate that~$\varphi$ is an isomorphism from~$G$ to~$H$. The~\emph{isomorphism type} of a graph~$G$ is the class of graphs isomorphic to~$G$.
An \emph{automorphism} of a graph $G$ is an isomorphism from~$G$ to itself. By $\Aut(G)$ we denote the group of automorphisms of $G$.
A graph $G$ is \emph{rigid} (or \emph{asymmetric}) if its automorphism group $\Aut(G)$ is trivial, that is, the only automorphism of $G$ is the identity map.

A vertex coloring of a graph~$G$ is a map~$c\colon V(G) \rightarrow \mathcal{C}$ into some set of colors~$\mathcal{C}$.
Mostly, we will use vertex colorings into the natural numbers~$c\colon V(G) \rightarrow  \{1,\ldots,n\}\eqqcolon [n]$. 

Isomorphisms between two colored graphs~$(G,c)$ and~$(G',c')$ are required to preserve vertex colors. Slightly abusing notation we will sometimes not differentiate between~$G$ and~$(G,c)$, if the coloring is apparent from context.

We will also consider vertex colored graphs with a distinguished sequence of not necessarily distinct vertices~$(G,c,\bar{v})$, where~$\bar v\in V^{\ell}$ for some~$\ell \in \mathbb{N}$.
For a tuple $\bar v \in V^{\ell}$ we let $|\bar v| = \ell$ be the length of the tuple.
Two such graphs with distinguished sequences~$(G,c,(v_1,\ldots,v_t))$ and~$(G',c',(v'_1,\ldots,v'_{t'}))$ are isomorphic if~$t= t'$ and there is an isomorphism~$\varphi$ from~$G$ to~$G'$ preserving vertex colors and satisfying~$\varphi(v_i) = v'_{i}$ for all~$i\in \{1,\ldots,t\}$.
In analogy to the definition above we write~$(G,c,(v_1,\ldots,v_t)) \stackrel{\varphi}{\cong} (G',c',(v'_1,\ldots,v'_{t'}))$. 

\subsection{The Weisfeiler-Leman algorithm}

The $k$-dimensional Weisfeiler-Leman algorithm is a procedure that, given a graph~$G$ and a coloring of the~$k$-tuples of the vertices, computes an isomorphism-invariant refinement of the coloring.
Let~$\chi,\chi'\colon V^k \rightarrow \mathcal{C}$ be colorings of the~$k$-tuples of vertices of~$G$, where~$\mathcal{C}$ is some set of colors. 
We say $\chi$ \emph{refines} $\chi'$ ($\chi \preceq \chi'$) if for all $\bar{v},\bar{w} \in V^k$ we have \[\chi(\bar{v}) = \chi(\bar{w}) \;\Rightarrow\; \chi'(\bar{v}) = \chi'(\bar{w}).\]

Let $(G,c)$ be a colored graph (where $c \colon V(G) \rightarrow \{1,\dots,n\}$ is a coloring of the vertices) and let $k \geq 1$ be some integer.
We set $\chi^G_0 \colon V^k \rightarrow \mathcal{C}$ to be the coloring, where each $k$-tuple is colored by the isomorphism type of its underlying induced ordered subgraph.
More precisely, we define $\chi^G_0$ in such a way that $\chi^G_0(v_1,\dots,v_k) = \chi^G_0(w_1,\dots,w_k)$ if and only if for all $i \in \{1,\dots,k\}$ it holds that $c(v_i) = c(w_i)$ and for all $i,j \in \{1,\dots,k\}$ we have $v_i = v_j \Leftrightarrow w_i = w_j$ and $\{v_i,v_j\} \in E(G) \Leftrightarrow \{w_i,w_j\} \in E(G)$.
For $k > 1$ we recursively define for~$G$ the coloring~$\chi^G_{i+1}$ by setting~$\chi^G_{i+1}(v_1, \dots, v_k) \coloneqq (\chi^G_{i} (v_1, \dots, v_k); \mathcal{M})$, where~$\mathcal{M}$ is the multiset defined as
\[\big\{\!\!\big\{\!\big(\!\chi^G_{i}(w,v_2, \dots, v_k),\chi^G_{i}(v_1,w, \dots, v_k), \dots, \chi^G_{i}(v_1, \dots, v_{k-1}, w)\big) \! \mid \! w\in V     \!\big\}\!\!\big\}.\]
For~$k=1$ the definition is analogous but the multiset is defined as~${\mathcal{M}}:=  \{\!\! \{ \chi^G_i(w) \mid w\in N(v_1) \}\!\! \}$ i.e., iterating only over neighbors of~$v_1$.

By definition, every coloring~$\chi^G_{i+1}$ induces a refinement of the partition of the~$k$-tuples of the graph~$G$ with coloring~$\chi^G_{i}$. Thus, there is some minimal~$i$ such that the partition induced by the coloring~$\chi^G_{i+1}$ is not strictly finer than the one induced by the coloring~$\chi^G_i$ on~$G$. For this minimal~$i$, we call the coloring~$\chi^G_i$ the \emph{stable} coloring of~$G$ and denote it by~$\chi^G$. 

For~$k \in \mathbb{N}$, the \emph{$k$-dimensional Weisfeiler-Leman algorithm} takes as input a colored graph~$(G, c)$ and returns the colored graph~$(G, \chi^G)$. For two colored graphs~$(G,c)$ and~$(G',c')$, we say that the~$k$-dimensional Weisfeiler-Leman algorithm \emph{distinguishes}~$G$ and~$G'$ with respect to the initial colorings~$c$ and~$c'$ if there is some color~$C$ such that the sets
$\{\bar{v} \mid \bar{v} \in V^k(G), \chi^G(\bar{v}) = C\}$ and $\{\bar{w} \mid \bar{w} \in V^k(G'), \chi^{G'}(\bar{w}) = C\}$
have different cardinalities. We write~$G \simeq_k H$ if the $k$-dimensional Weisfeiler-Leman algorithm does not distinguish between~$G$ and~$H$.

We extend the definition to vertex-colored graphs with distinguished vertices.
Let~$G$ and~$H$ be graphs with vertex colorings~$c_G$ and~$c_H$ and let~$(x_1,\ldots,x_t)\in V(G)^t$ and~$(y_1,\ldots,y_t)\in V(H)^t$ be sequences of vertices.
Define~$\widehat{c_G}$ as the coloring given by~$\widehat{c_G}(v) = (c_G(v); \{i\in \{1,\ldots,t\} \mid v= x_i\})$. We call this the coloring obtained from~$c_G$ by individualizing~$(x_1,\ldots,x_t)$. Similarly we define~$\widehat{c_H}$ as~$(c_H(v); \{i\in \{1,\ldots,t\} \mid v= y_i\})$.
Then we say that~$(G,c_G,(x_1,\ldots,x_t))$ is not distinguished from~$(H,c_H,(y_1,\ldots,y_t))$, in symbols~$(G,c_G,(x_1,\ldots,x_t)) \simeq_k (H,c_H,(y_1,\ldots,y_t))$, if~$G$ and~$H$ are not distinguished by the~$k$-dimensional Weisfeiler-Leman algorithm with respect to the initial colorings~$\widehat{c_G}$ and~$\widehat{c_H}$.
We denote by~$\WL_k(G,c,\bar{v})$ the vertex coloring that is induced by the stable coloring with respect to the initial coloring~$\widehat{c_G}$, that is, $(\WL_k(G,c,\bar{v}))(w) = \chi^{G}(w,\dots,w)$ where $\chi^{G}$ is the stable coloring of the $k$-tuples with respect to the initial coloring~$\widehat{c_G}$.

There is a close connection between the Weisfeiler-Leman algorithm and fixed-point logic with counting. In fact the stable coloring computed by~$k$-dimensional Weisfeiler-Leman comprehensively captures the information that can be obtained in fixed-point logic with counting using at most~$k+1$ variables. We refer to~\cite{cfi,IL90} for more details. 

\paragraph{Pebble Games.} We will not require details about the information computed by the Weisfeiler-Leman algorithm and rather use the following pebble game that is known to capture the same information. 
Let $k \in \mathbb{N}$ be a fixed number. For graphs $G,H$ on the same number of vertices and with vertex colorings~$c_G$ and~$c_H$, respectively, we define the bijective $k$-pebble game on $G$ and $H$ as follows:
\begin{itemize}
 \item The game has two players called Spoiler and Duplicator
 \item The game proceeds in rounds. Each round is associated with a pair of positions $(\bar v,\bar w)$ with~$\bar v \in (V(G)\cup \{\bot\})^k$ and~$\bar w \in (V(G)\cup \{\bot\})^k$. 
 \item The initial position of the game is $((\bot,\dots,\bot),(\bot,\dots,\bot))$.
 \item Each round consists of the following steps. Suppose the current position of the game is $((v_1,\ldots,v_k),(w_1,\ldots,w_k))$.
  \begin{itemize}
   \item[(S)] Spoiler chooses some $i \in [k]$.
   \item[(D)] Duplicator picks a bijection $f\colon V(G) \rightarrow V(H)$.
   \item[(S)] Spoiler chooses $v \in V(G)$ and sets $w = f(v)$.
  \end{itemize}
  The new position is then the pair consisting of $\bar v' = ( v_1,\dots, v_{i-1},v, v_{i-1},\dots, v_k)$ and $\bar w' = ( w_1,\dots, w_{i-1},w, w_{i-1},\dots, w_k)$.
  
  \item Spoiler wins the game if for the current position~$((v_1,\ldots,v_k),(w_1,\ldots,w_k))$ the induced graphs are not isomorphic.
  More precisely, Spoiler wins if there is an~$i\in \{1,\ldots,k\}$ such that~$v_i=\bot \notleftright w_i = \bot$ or~$c_G(v_i)\neq c_H(w_i)$ or there are~$i,j\in \{1,\ldots,k\}$ such that~$v_i = v_j\notleftright w_i =w_j$ or~$\{v_i,v_j\} \in E(G)\notleftright \{w_i,w_j\} \in E(H)$.
  If the play never ends Duplicator wins.
\end{itemize}

We say that Spoiler (respectively Duplicator) wins the bijective $k$-pebble game $\BP_k(G,H)$ if Spoiler (respectively Duplicator) has a winning strategy for the game.

\begin{theorem}[cf.\ \cite{cfi,IL90}]
 Let $G, H$ be two graphs.
 Then $G \simeq_k H$ if and only if Duplicator wins the pebble game $\BP_{k+1}(G,H)$.
\end{theorem}

\section{Individualization-refinement algorithms}

An extensive description of the paradigm of individualization-refinement algorithms is given in \cite{mckay}. These algorithms capture information about the structure of a graph by coloring the vertices. An initially uniform coloring is first refined in an isomorphism-invariant manner as follows.

A \emph{refinement operator} is an isomorphism-invariant function $\refine$ that takes a graph $G$, a coloring~$c$ and a sequence $\bar{v} = (v_1,\dots,v_\ell) \in V^{\ell}$ and outputs a coloring $\refine(G,c,\bar{v}) \preceq c$ such that $v_i$ has a unique color for every $i \in [\ell]$. 
In this context isomorphism-invariant means that $(G,c,\bar{v})  \stackrel{\varphi}{\cong} (G',c',\bar{v}')$ implies~$(G,\refine(G,c,\bar{v})) \stackrel{\varphi}{\cong} (G',\refine(G',c',\bar{v}'))$. A typical choice for such a refinement would be the 1-dimensional Weisfeiler-Leman algorithm described above, where the vertices in~$\bar{v}$ are artificially given special colors.

A vertex with a unique color is called a \emph{singleton} and a coloring is called \emph{discrete} if all vertices are singletons. Due to the isomorphism-invariance, every isomorphism must preserve the refined colors. Thus, 
in case the refinement operator produces a discrete coloring on a graph~$G$, it is trivial to check whether this graph is isomorphic to another graph~$G'$. Indeed, the refinement of~$G'$ must also be discrete and there is at most one color preserving bijection between the vertex sets which can be trivially checked for being an isomorphism.
However, if the coloring of~$G$ is not discrete we need to do more work. In this case we select a color class, usually called a \emph{cell}, and then individualize a single vertex from the class.  Here \emph{individualization} means to refine the coloring by making the vertex a singleton. 
Since such an operation is not necessarily isomorphism-invariant, we branch over all choices of this vertex within the chosen cell. 
To the coloring with the newly individualized vertex, we apply the refinement operator again and proceed in a recursive fashion. To explain this in more detail we first need to clarify how the cell is chosen. 

Let $G=(V,E)$ be a graph and $c\colon V(G) \rightarrow [n]$ be a coloring of the vertices.
A \emph{cell selector} is an isomorphism-invariant function $\sel$ which takes as input a graph $G$ and a coloring $c$ and either outputs $\sel(G,c) \in [n]$ with $|c^{-1}(\sel(G,c))| \geq 2$ if such a color exists or $\sel(G,c) = \perp$ otherwise. 
In this context isomorphism-invariant means that $(G,c)\cong (G',c')$ implies~$\sel(G,c) =\sel(G',c')$. The performance of an individualization-refinement algorithm can drastically depend on the cell selection strategy. A typical strategy would be to take the first class of smallest size.

Let $G = (V,E)$ be a graph with an initial coloring $c_0\colon V(G) \rightarrow [n]$.
Let $\sel$ be a cell selector and $\refine$ a refinement operator.
Inductively define the search tree $\mathcal{T}^{\refine,\sel}(G,c_0)$ as follows.
The root of the tree is labeled with the empty sequence $\varepsilon$.
Let $\bar{v} = (v_1,\dots,v_\ell)$ be a node of the search tree.
Let $c = \refine(G,c_0,\bar{v})$ be the coloring computed by the refinement operator for the current sequence and let $i = \sel(G,c)$ be the color selected by the cell selector.
If $i = \perp$ then $\bar{v}$ is a leaf of the search tree and the coloring $c$ is discrete.
Otherwise, for each $w \in c^{-1}(i)$, there is child node labeled with $(v_1,\dots,v_\ell,w)$. The vertices of the search tree are referred to as \emph{nodes} and we identify them with the sequence of vertices they are labeled with.

\paragraph{Pruning with invariants.}
Together a cell selector and a refinement operator are sufficient to build a correct isomorphism test. 
Indeed, two graphs are isomorphic if and only if they have isomorphic leaves in their search tree. For these leaves, due to having a discrete coloring, isomorphism is trivial to check.
However, there are two further ingredients that are crucial for the efficiency of practical individualization-refinement algorithms. These are the use of node invariants and the exploitation of automorphisms. 
Let $\Omega$ be a totally ordered set.
A \emph{node invariant} is an isomorphism-invariant function $\inv$ taking a graph $G$, a coloring $c$ and a sequence $\bar{v} = (v_1,\dots,v_\ell) \in V^{\ell}$ and outputs an element $\inv(G,c,\bar v) \in \Omega$ such that for all vertex sequences $\bar v, \bar v' \in V^{*}$ of equal length~$\ell$
\begin{itemize}
 \item[(i)] if $\inv(G,c,(v_1,\dots,v_\ell)) < \inv(G,c,(v'_1,\dots,v'_\ell))$ then it also holds for all $w,w'\in V$ that $\inv(G,c,(v_1,\dots,v_\ell,w)) < \inv(G,c,(v'_1,\dots,v'_\ell,w'))$ 
and
 \item[(ii)] if $\refine(G,c,\bar v)$ and $\refine(G,c,\bar v')$ are discrete and $\inv(G,c,\bar v) = \inv(G,c,\bar v')$ then $(G,c,\bar v) \cong (G,c,\bar v')$.
\end{itemize}
Here isomorphism-invariant means that $(G,c,\bar{v})  \stackrel{}{\cong} (G',c',\bar{v}')$ implies~$\inv(G,c,\bar{v}) = \inv(G',c',\bar{v}')$.

Let $\inv$ be a node invariant and define
\[\mathcal{I} = \{\bar v \in V(\mathcal{T}^{\refine,\sel}(G,c_0)) \mid \nexists\, \bar v' \in V(\mathcal{T}^{\refine,\sel}(G,c_0))\colon |\bar v| = |\bar v'| \wedge \left(\inv(G,c_0,\bar v') < \inv(G,c_0,\bar v)\right)\}.\]
Finally, define the search tree $\mathcal{T}_{\inv}^{\refine,\sel}(G,c_0) \coloneqq  \mathcal{T}^{\refine,\sel}(G,c_0)[\mathcal{I}]$ as the subtree induced by the node set $\mathcal{I}$.
Observe that Property (i) implies that  $\mathcal{T}_{\inv}^{\refine,\sel}(G,c_0)$ is indeed a tree.  By using the invariant we thus cut off the parts of the search tree that do not have a nodes that are minimal among all nodes on their level. However, due to isomorphism invariance, the property that two graphs 
are isomorphic if and only if they have isomorphic leaves remains.

The use of an invariant also makes it easy to define a canonical labeling using a leaf for which the invariant is smallest.
For the purpose of obtaining our lower bounds we will not require detailed information on the concept of a canonical labeling and rather refer to~\cite{mckay}.

\paragraph{Pruning with automorphisms.} The second essential ingredient needed for the practicality of individualization-refinement algorithms is the exploitation of automorphisms. Indeed, if for two nodes in~$\mathcal{T}_{\inv}^{\refine,\sel}(G)$, labeled with~$\bar{v}$ and~$\bar{v}'$, respectively,   we have~$(G,c,\bar{v})  \stackrel{}{\cong} (G',c',\bar{v}')$ then it is sufficient to explore only one of the subtrees corresponding to the two nodes (we refer to \cite{mckay} for correctness arguments). 
Thus automorphisms that are detected by the algorithm can be used to cut off further parts of the search tree.
An efficient strategy for the detection of automorphisms is an essential part of individualization-refinement algorithms and here the various packages differ drastically (see~\cite{mckay}). Making our lower bounds only stronger, in this paper we take the following standpoint.
We will assume that all automorphisms of the input graph are provided to the algorithm in the beginning at no cost. In fact the following lower bound will be sufficient for our purposes.

\begin{proposition}\label{prop:size:of:tree:bounds:run:time}
The running time of an individualization-refinement algorithm with cell selector~$\sel$, refinement operator~$\refine$ and invariant~$\inv$ on a graph~$G$ is bounded from below by~$|\mathcal{T}_{\inv}^{\refine,\sel}(G)|/ |\Aut(G)|$.
\end{proposition}

The program nauty has an extensive selection of refinement operators/invariants that can be activated via switches.
To name a couple, there are various options to count for each vertex the number of vertices reachable by paths with vertex colors of a certain type, options to count substructures such as triangles, quadrangles, cliques up to size 10, independent sets up to size 10, and even options to count the number of Fano planes (the projective plane with 7 points and 7 lines). Finally, the user can implement their own invariant via a provided interface.

Our goal is to make a comprehensive statement about individualization-refinement algorithms independent of the choices for~$\sel$,~$\refine$, and~$\inv$. 
However, there is an intrinsic limitation here. For example a complete invariant that can distinguish any two non-isomorphic graphs would yield a polynomial-size search tree.
Likewise would a refinement operator that refines every coloring into the orbit partition under the automorphism group.
However, we do not know how to compute these two examples efficiently.
In fact computing either of these is at least as hard as the isomorphism problem itself. 

Of course it is nonsensical to allow that an individualization-refinement algorithm uses a subroutine that already solves the graph isomorphism problem.
It becomes apparent that there is a limitation to the operators we can allow. However, within this limitation we strive to be as general as possible.
With this in mind, throughout this paper we require that the information computed by refinement operators, invariants and cell selectors can be captured by a fixed dimension of the Weisfeiler-Leman algorithm.
This is the case for all available choices in all the practical algorithms. In what follows, we describe the requirement more formally.

We say a cell selector $\sel$ is \emph{$k$-realizable} if $\sel(G,c,\bar v) = \sel(G',c',\bar v')$ whenever $(G,c,\bar v) \simeq_k (G',c',\bar v')$.
Similarly a node invariant is \emph{$k$-realizable} if $\inv(G,c,\bar v) = \inv(G',c',\bar v')$
whenever $(G,c,\bar v) \simeq_k (G',c',\bar v')$.
Intuitively this means that whenever the $k$-dimensional Weisfeiler-Leman algorithm cannot distinguish between the graphs associated with two nodes of the refinement tree then the cell selector and the node invariant have to behave in the same way on both nodes.
Finally a refinement operator is \emph{$k$-realizable} if~$\WL_k(G,c,\bar{v})\preceq \refine(G,c,\bar{v}) $ holds for all triples~$(G,c,\bar{v})$.

We want to stress the fact that all the operators used in all practical implementations (e.g.\ nauty/traces, bliss, conauto, etc.) are $k$-realizable for some small constant $k \in \mathbb{N}$. In fact, from a theoretical point of view it would always be better to directly use the Weisfeiler-Leman algorithm as a refinement operator, since it is polynomial-time computable. Let us remark that the only reason why the individualization-refinement algorithms do not do this is the excessive running time and space consumption.

Based on these definitions we can now formulate our main result, which implies an exponential lower bound for individualization-refinement algorithms within the framework.

\begin{theorem}
 \label{thm:main}
 For every constant $k \in \mathbb{N}$ there is a family of rigid graphs $(G_n)_{n \in \mathbb{N}}$ with $|V(G_n)| \leq n$ such that for every $k$-realizable cell selector $\sel$, every $k$-realizable refinement operator $\refine$, and every $k$-realizable node invariant $\inv$ it holds that
 \begin{equation}
  |\mathcal{T}_{\inv}^{\refine,\sel}(G_n)| \in 2^{\Omega(n)}.
 \end{equation}
\end{theorem}

Together with Proposition~\ref{prop:size:of:tree:bounds:run:time} this implies exponential lower bounds on the running time of individualization-refinement algorithms.

For the theorem we construct graphs which have large search trees. To prove that these search trees are large, we use the following lemma stating that for every two tuples $\bar v, \bar w \in V^{\ell}$ for which $(G,\bar v) \simeq_k (G,\bar w)$ holds, either both tuples are nodes in the search or neither of them is.

\begin{lemma}
 \label{la:tree-nodes-from-equivalent-tuples}
 Suppose $k \in \mathbb{N}$ and let $\sel$ be a $k$-realizable cell selector, $\inv$ a $k$-realizable node invariant and $\refine$ a $k$-realizable refinement operator.
 Furthermore, let $G$ be a graph and suppose $\bar v \in V(\mathcal{T}_{\inv}^{\refine,\sel}(G))$.
 Let $m = |\bar v|$.
 Then $\bar w \in V(\mathcal{T}_{\inv}^{\refine,\sel}(G))$ for every $\bar w \in V(G)^{m}$ with $(G,\bar v) \simeq_k (G,\bar w)$.
\end{lemma}

\begin{proof}
 We prove the statement by induction on $m \in \mathbb{N}$.
 For $m = 0$ the statement trivially holds since $\bar v = \bar w = \varepsilon$.
 So suppose $m > 0$.
 Let $\bar v' = (\bar v_1,\dots,\bar v_{m-1})$ be the tuple obtained from $\bar v$ by deleting the last entry and similarly define $\bar w' = (\bar w_1,\dots,\bar w_{m-1})$.
 Clearly, $(G,\bar v') \simeq_k (G,\bar w')$ and $\bar v' \in V(\mathcal{T}_{\inv}^{\refine,\sel}(G))$.
 So by induction hypothesis it follows that $\bar w' \in V(\mathcal{T}_{\inv}^{\refine,\sel}(G))$.
 Let $c_1 = \refine(G,\bar v')$ and $c_2 = \refine(G,\bar w')$.
 Because $\sel$ and~$\refine$ are $k$-realizable we get that $\sel(G,c_1) = \sel(G,c_2) \eqqcolon i$.
 So $\bar v_m \in c_1^{-1}(i)$ and also $\bar w_m \in c_2^{-1}(i)$ since $(G,\bar v) \simeq_k (G,\bar w)$.
 Thus, $\bar w \in V(\mathcal{T}^{\refine,\sel}(G))$.
 Furthermore, $\inv(G,\bar v) = \inv(G,\bar w)$ which implies that $\bar w \in V(\mathcal{T}_{\inv}^{\refine,\sel}(G))$.
\end{proof}

In the light of the lemma, to prove our lower bound it suffices to construct a graph~$G$ whose search tree has a node $\bar v$ with an exponential number of equivalent tuples.
To argue the existence of these we, roughly proceed in two steps.
First, we show that to obtain a discrete partition we have to individualize a linear number of vertices, in other words, we show that the search tree is of linear height. 
Thus, there is a node $\bar v$ in the corresponding search such that $|\bar v|$ is linear in the number of vertices of $G$.
Then, in a second step, we show that if~$|\bar v|$ is sufficiently large, there are exponentially many equivalent tuples.
To find such equivalent tuples we prove a limitation of the effect of the $k$-dimensional Weisfeiler-Leman algorithm after individualizing the vertices from~$\bar v$. Intuitively we identify a subgraph containing~$\bar v$ which encapsulates the effect of the Weisfeiler-Leman algorithm.
We then exploit the existence of many automorphisms of this subgraph to demonstrate the existence of many equivalent tuples.

\section{The multipede construction}

We describe a construction of a graph~$R(G)$ from a bipartite base graph~$G$. The construction is a combination of the Cai-Fürer-Immerman construction~\cite{cfi} giving graphs with large Weisfeiler-Leman dimension and the construction of Gurevich and Shelah of multipedes \cite{DBLP:journals/jsyml/GurevichS96} that yields rigid structures with such properties.

\paragraph{The Cai-Fürer-Immerman gadget.} For a non-empty finite set $S$ we define the CFI gadget $X_S$ to be the following graph.
For each $w \in S$ there are vertices $a_w$ and $b_w$ and for every $A \subseteq S$ with $|A|$ even there is a vertex $m_A$.
For every $A \subseteq S$ with $|A|$ even there are edges $\{a_w,m_A\} \in E(X_S)$ for all $w \in A$ and $\{b_w,m_A\} \in E(X_S)$ for all $w \in S\setminus A$.
As an example the graph $X_3 := X_{[3]}$ is depicted in Figure \ref{fig:cfi-gadget}.
The graph is colored so that~$\{a_w,b_w\}$ forms a color class for each~$w$ and so that~$\{m_A\mid A \subseteq S \text{ and } |A| \text{ even}\}$ forms a color class.

Let $X \subseteq S$ and $\gamma \in \Aut(X_S)$.
We say that $\gamma$ \emph{swaps} exactly the pairs of $X$ if $\gamma(a_w) = b_w$ for $w \in X$ and $\gamma(a_w) = a_w$ for $w \in S \setminus X$.

\begin{lemma}[\cite{cfi}]
 \label{la:cfi-gadget}
 Let $X \subseteq S$.
 Then there is an automorphism $\gamma \in \Aut(X_S)$ swapping exactly the pairs of $X$ if and only if $|X|$ is even.
 Additionally, if such an automorphism exists, it is unique.
\end{lemma}

\begin{figure}
 \centering
 \begin{tikzpicture}
  \node[style=normalvertex,fill=red,label=right:{$a_1$}] (a1) at (6,2) {};
  \node[style=normalvertex,fill=red,label=right:{$b_1$}] (b1) at (6,3) {};
  
  \node[style=normalvertex,fill=blue,label=left:{$a_2$}] (a2) at (0,0) {};
  \node[style=normalvertex,fill=blue,label=left:{$b_2$}] (b2) at (0,1) {};
  
  \node[style=normalvertex,fill=green,label=left:{$a_3$}] (a3) at (0,4) {};
  \node[style=normalvertex,fill=green,label=left:{$b_3$}] (b3) at (0,5) {};
  
  \node[style=normalvertex,label=85:{{\small $m_\emptyset$}}] (v1) at (3,4) {};
  \node[style=normalvertex,label=85:{{\small $m_{\{2,3\}}$}}] (v2) at (3,3) {};
  \node[style=normalvertex,label=85:{{\small $m_{\{1,3\}}$}}] (v3) at (3,2) {};
  \node[style=normalvertex,label=85:{{\small $m_{\{1,2\}}$}}] (v4) at (3,1) {};
  
  \path
   (v1) edge (b1)
   (v1) edge (b2)
   (v1) edge (b3)
   (v2) edge (b1)
   (v2) edge (a2)
   (v2) edge (a3)
   (v3) edge (a1)
   (v3) edge (b2)
   (v3) edge (a3)
   (v4) edge (a1)
   (v4) edge (a2)
   (v4) edge (b3);
  
 \end{tikzpicture}
 \caption{Cai-F\"{u}rer-Immerman gadget $X_3$}
 \label{fig:cfi-gadget}
\end{figure}
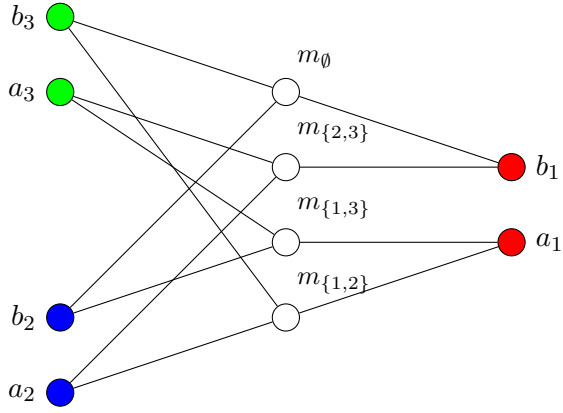

\paragraph{The multipede graphs.}

Let $G=(V,W,E)$ be a bipartite graph.
We define the multipede graph construction as follows.
We replace every $w \in W$ by two vertices $a(w)$ and $b(w)$.
For each $w \in W$ let $F(w) = \{a(w),b(w)\}$ and for $X \subseteq W$ define $F(X) = \bigcup_{w \in X} F(w)$.
We then replace every $v \in V$ by the CFI gadget $X_{N(v)}$ and identify the vertices $a_w$ and $b_w$ with $a(w)$ and $b(w)$ for all $w \in N(v)$.
The middle vertices will be denoted by $m_A(v)$ for $A \subseteq N(v)$ with $|A|$ even and the set of all middle vertices of $v \in V$ is denoted by $M(v)$.
For $Y \subseteq V$ we define $M(Y) = \bigcup_{v \in Y}M(v)$.
We also define a vertex coloring for the graphs, however we only specify the color classes, since the actual names of the  colors will be irrelevant to us.
In this manner, we require that for each $w \in W$, the pair of vertices $F(w)$ forms a color class, and for each $v \in V$ the set of corresponding middle vertices $M(v)$ forms a color class.
The resulting (vertex colored) graph will be denoted by $R(G)$. An example of this construction is shown in Figure~\ref{fig:bipartite:cfi:constr}.

For $I \subseteq W$ we further define the graph $R^{I}(G)$ similar to $R(G)$ but refine the coloring so that for each~$w\in I$ both~$\{a(w)\}$ and~$\{b(w)\}$ form a color class.
Hence, $R(G) = R^{\emptyset}(G)$.

\begin{figure}
\centering
\begin{tikzpicture}
\tikzstyle{normalvertex}=[circle, draw]
		\node [style=normalvertex] (0) at (-11, -2) {};
		\node [style=normalvertex] (1) at (-8, -2) {};
		\node [style=normalvertex] (2) at (-4, -2) {};
		\node [style=normalvertex] (3) at (-12, -0) {};
		\node [style=normalvertex] (4) at (-10, -0) {};
		\node [style=normalvertex] (5) at (-8, -0) {};
		\node [style=normalvertex] (6) at (-6, -0) {};
		\node [style=normalvertex] (7) at (-2, -0) {};
		\node [style=normalvertex] (8) at (-4, -0) {};
		\node [style=normalvertex,fill=red] (9) at (-11.75, -4) {};
		\node [style=normalvertex,fill=red] (10) at (-12.25, -4) {};
		\node [style=normalvertex,fill=green] (11) at (-10.25, -4) {};
		\node [style=normalvertex,fill=green] (12) at (-9.75, -4) {};
		\node [style=normalvertex,fill=blue] (13) at (-8.25, -4) {};
		\node [style=normalvertex,fill=blue] (14) at (-7.75, -4) {};
		\node [style=normalvertex] (15) at (-11.25, -6) {};
		\node [style=normalvertex] (16) at (-11.75, -6) {};
		\node [style=normalvertex] (17) at (-10.75, -6) {};
		\node [style=normalvertex] (18) at (-10.25, -6) {};
		\node [style=normalvertex,fill=gray!80] (19) at (-8.75, -6) {};
		\node [style=normalvertex,fill=gray!80] (20) at (-8.25, -6) {};
		\node [style=normalvertex,fill=gray!80] (21) at (-7.75, -6) {};
		\node [style=normalvertex,fill=gray!80] (22) at (-7.25, -6) {};
		\node [style=normalvertex,fill=orange] (23) at (-6.25, -4) {};
		\node [style=normalvertex,fill=orange] (24) at (-5.75, -4) {};
		\node [style=normalvertex,fill=cyan] (25) at (-4.25, -4) {};
		\node [style=normalvertex,fill=cyan] (26) at (-3.75, -4) {};
		\node [style=normalvertex,fill=violet] (27) at (-2.25, -4) {};
		\node [style=normalvertex,fill=violet] (28) at (-1.75, -4) {};
		\node [style=normalvertex,fill=black!80] (29) at (-4.75, -6) {};
		\node [style=normalvertex,fill=black!80] (30) at (-4.25, -6) {};
		\node [style=normalvertex,fill=black!80] (31) at (-3.75, -6) {};
		\node [style=normalvertex,fill=black!80] (32) at (-3.25, -6) {};
	
		\draw (3) to (0);
		\draw (0) to (4);
		\draw (0) to (5);
		\draw (4) to (1);
		\draw (1) to (5);
		\draw (1) to (6);
		\draw (6) to (2);
		\draw (2) to (7);
		\draw (8) to (2);
		\draw (10) to (16);
		\draw (10) to (15);
		\draw (9) to (17);
		\draw (9) to (18);
		\draw (11) to (16);
		\draw (11) to (17);
		\draw (12) to (15);
		\draw (12) to (18);
		\draw (13) to (16);
		\draw (13) to (18);
		\draw (14) to (15);
		\draw (17) to (14);
		\draw (19) to (11);
		\draw (11) to (20);
		\draw (21) to (12);
		\draw (12) to (22);
		\draw (13) to (19);
		\draw (13) to (21);
		\draw (14) to (20);
		\draw (14) to (22);
		\draw (23) to (19);
		\draw (23) to (22);
		\draw (24) to (20);
		\draw (24) to (21);
		\draw (23) to (29);
		\draw (23) to (30);
		\draw (24) to (31);
		\draw (24) to (32);
		\draw (25) to (29);
		\draw (25) to (31);
		\draw (26) to (30);
		\draw (26) to (32);
		\draw (27) to (29);
		\draw (27) to (32);
		\draw (28) to (30);
		\draw (28) to (31);
	
\node at (-13,0) {$W$};

\node at (-13,-2) {$V$};

\node at (-14.5,-1) {$G$};
\node at (-14.5,-5) {$R(G)$};

\node at (-11,-2.5) {$v_1$};

\node at (-8,-2.5) {$v_2$};

\node at (-4,-2.5) {$v_3$};

\foreach \x in {1,...,6}{
\node at (-14+2*\x,0.5) {$w_\x$};
}

\foreach \x in {0,3,7}{
\draw[dashed]  (-12.1+\x,-5.65) rectangle (-9.9+\x,-6.35);
}
\foreach \x in {1,...,6}{

\node at (-14+2*\x,-3.5) {\small{$a(w_\x)\, b(w_\x)$}};
}
\end{tikzpicture}
 \caption[The bipartite~CFI-construction]{The figure depicts a base graph~$G$ on the top and the graph~$R(G)$ obtained by applying the multipede construction on the bottom.}
 \label{fig:bipartite:cfi:constr}
\end{figure}
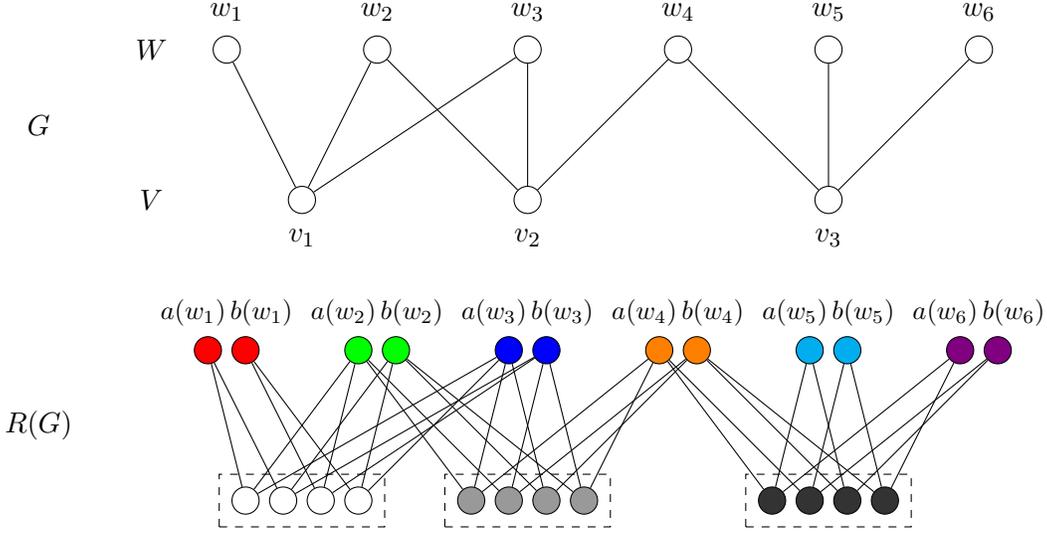

Since we are aiming to construct rigid graphs we start by identifying properties of $G$ that correspond to $R(G)$ having few automorphisms.

\begin{definition}
 Let $G=(V,W,E)$ be a bipartite graph.
 We say $G$ is \emph{odd} if for every $\emptyset \neq X \subseteq W$ there exists some $v \in V$ such that $|N(v) \cap X|$ is odd.
\end{definition}

\begin{lemma}
 \label{la:odd-is-rigid}
 Let $G=(V,W,E)$ be an odd bipartite graph.
 Then $R^{I}(G)$ is rigid for every $I \subseteq W$.
\end{lemma}

\begin{proof}
 Let $\gamma \in \Aut(R^{I}(G))$ be an automorphism.
 Due to the coloring of the vertices, we know $\gamma$ stabilizes every set $F(w)$ for $w \in W$ and every set $M(v)$ for $v \in V$.
 
 Let $X = \{w \in W \mid \gamma(a(w)) = b(w)\}$.
 Suppose towards a contradiction that $X \neq \emptyset$.
 Since $G$ is odd there is some $v \in V$ such that $|N(v) \cap X|$ is odd.
 Then $\gamma$ restricts to an automorphism of the gadget $X_{N(v)}$ swapping an odd number of the outer pairs.
 This contradicts the properties of CFI gadgets (cf.\ Lemma \ref{la:cfi-gadget}).
 
 So $X = \emptyset$ and thus $\gamma(a(w)) = a(w)$ for all $w \in W$.
 From this it easily follows that $\gamma$ is the identity mapping.
\end{proof}

For a bipartite graph $G = (V,W,E)$ let $A_G \in \mathbb{F}_2^{V \times W}$ be the matrix with $A_{v,w} = 1$ if and only if $\{v,w\} \in E(G)$. We denote by $A^{T}$ the transpose and by $\rk(A)$ the~$\mathbb{F}_2$-rank of a matrix $A$.

\begin{lemma}
 \label{la:size-automorphism-group}
 Let $G=(V,W,E)$ be a bipartite graph.
 Then $|\Aut(R(G))| = 2^{|W| - \rk(A_G)}$.
\end{lemma}

\begin{proof}
 For a matrix $A \in \mathbb{F}_2^{m \times n}$ we define the set $\Sol(A) = \{x \in \mathbb{F}_2^{n} \mid Ax = 0\}$.
 To show the lemma it suffices to argue that
 $|\Aut(R(G))| = |\Sol(A_G)|$.
 
 For $\gamma \in \Aut(G)$ we define the vector $x_\gamma \in \mathbb{F}_2^{W}$ by setting $(x_\gamma)_w = 1$ if and only if $\gamma(a(w)) = b(w)$.
 Observe that the mapping $\gamma \mapsto x_\gamma$ is injective. Furthermore $A_Gx_\gamma = 0$ since for each $v \in V$ the automorphism $\gamma$ swaps an even number of neighbors of $v$.
  
 For the backward direction let $x \in \Sol(A_G)$.
 Then, for each $v \in V$, the set $\{w \in N(v) \mid x_w = 1\}$ has even cardinality.
 Thus, by the properties of the CFI-gadgets (cf.\ Lemma \ref{la:cfi-gadget}), there is a unique automorphism $\gamma \in \Aut(R(G))$ that swaps exactly those pairs $(a(w),b(w))$ for which $x_w = 1$.
\end{proof}

The arguments show that a bipartite graph $G=(V,W,E)$ is odd if and only if~$\rk(A_G)=|W|$.

\begin{corollary}
 \label{cor:reduce-size-odd}
 Let $G=(V,W,E)$ be an odd bipartite graph.
 Then there is some $V' \subseteq V$ with $|V'| \leq |W|$ such that the induced subgraph $G[V' \cup W]$ is odd.
\end{corollary}

We can also use~$\rk(A_G)$ to quantify how many vertices must be individualized to make~$G$ rigid.

\begin{lemma}
 \label{la:make-rigid}
 Let $G=(V,W,E)$ be a bipartite graph.
 Then there is $I \subseteq W$ with $|I| \leq |W| - \rk(A_G)$ such that $R^{I}(G)$ is rigid.
\end{lemma}

\begin{proof}
 Let $B = \{e_w \in \mathbb{F}_2^{W} \mid w \in W\}$ be the standard basis for $\mathbb{F}_2^{W}$ (that is, $(e_w)_u = 1$ if and only if $w= u$).
 Furthermore, for $v \in V$, let $(A_G)_v$ be the $v$-th row of $A_G$ and let $B_I \subseteq B$ be a minimal subset of $B$ such that $B_I \cup \{((A_G)_v)^{T} \mid v \in V\}$ spans the entire space $\mathbb{F}_2^{W}$.
 Finally, let $I = \{w \in W \mid e_w \in B_I\}$.
 Clearly, $|I| \leq |W| - \rk(A_G)$.

 We argue that  $R^{I}(G)$ is rigid.
 Let $\gamma \in \Aut(R^{I}(G))$ and let $x_\gamma \in \mathbb{F}_2^{W}$ be the vector obtained by setting $(x_\gamma)_w = 1$ if and only if $\gamma(a(w)) = b(w)$.
 Then $(e_w)^{T}x_\gamma = 0$ for all $w \in I$.
 Furthermore $(A_G)_vx_\gamma = 0$ for all $v \in V$ by the same argument as in the proof of Lemma \ref{la:size-automorphism-group}.
 Since $B_I \cup \{((A_G)_v)^{T} \mid v \in V\}$ spans the entire space $\mathbb{F}_2^{W}$ it follows by the standard linear algebra arguments that $x_\gamma = 0$.
 Thus $\gamma$ is the identity.
\end{proof}

\section{The Weisfeiler-Leman refinement and the $d$-closure}

We wish to understand the effect of the Weisfeiler-Leman refinement on graphs obtained with the construction from the previous section. To this end we define the~$d$-closure.

\begin{definition}
 Let $d \in \mathbb{N}$. For $X \subseteq W$ define the \emph{$d$-attractor} of~$X$ as \[\attr^{d}(X) \coloneqq X \cup \bigcup_{v \in V\colon |N(v) \setminus X| \leq d} N(v).\]
 A set $X \subseteq W$ is \emph{$d$-closed} if $X = \attr^{d}(X)$.
 The \emph{$d$-closure} of $X$ is the unique minimal superset which is $d$-closed, that is \[\cl_G^{d}(X) = \bigcap_{\substack{
 X' \supseteq X,\\ X' \text{ is $d$-closed}}}X'.\]
\end{definition}

As observed in~\cite{DBLP:journals/jsyml/GurevichS96} the 1-closure describes exactly the information the 1-dimensional Weisfeiler-Leman captures.

\begin{lemma}
 \label{la:color-refinement-on-bipartite-cfi}
 Let $G=(V,W,E)$ be a bipartite graph and suppose $I \subseteq W$.
 Then \[(R^{I}(G),a(w)) \not\simeq_1 (R^{I}(G), b(w)) \;\;\Leftrightarrow\;\; w \in \cl_G^{1}(I)\] for all $w \in W$.
\end{lemma}

\begin{proof}[Proof sketch]
 The backward direction follows by an inductive argument from the properties of the CFI gadgets.
 For the forward direction it is easy to check that the corresponding partition on $F(W)$ directly extends to a stable partition for the graph $R^{I}(G)$.
\end{proof}

Let $G=(V,W,E)$ be a bipartite graph. Slightly abusing notation, for a set $X \subseteq W$ define $N^{-1}(X) \coloneqq \{v \in V \mid N(v) \subseteq X\}$.
For $X \subseteq W$ define $R(G)[[X]] \coloneqq  R^{}(G)[F(X) \cup M(N^{-1}(X))]$ (for a graph $G$ and a set $X \subseteq V(G)$ the graph $G[X]$ is the induced subgraph of $G$ with vertex set $X$).

In this work we essentially argue that the previous lemma can be used to show that for a 1-closed set~$X\subseteq W$ and a sequence of vertices $\bar x = (x_1,\dots,x_m)$ from~$F(W)$, for every automorphism $\varphi \in \Aut(R(G)[[X]])$ we have $(R(G),\bar x) \simeq_1 (R(G), \varphi(\bar x))$.
The 1-closure thus gives us a method to find tuples that cannot be distinguished by the 1-dimension Weisfeiler-Leman algorithm.
However, we require such a statement also for higher dimensions.
Obtaining a similar statement characterizing the effect of the~$k$-dimensional Weisfeiler-Leman seems to be much more intricate and it is easy to see that the $d$-closure does not achieve this.
However, under some additional assumptions, we show that the forward direction of the previous lemma still holds and thus, the $d$-closure gives us a tool to control the effect of $k$-dimensional Weisfeiler-Leman which is sufficient for our purposes.

\begin{lemma}
 \label{la:equivalent-tuples-from-automorphisms}
 Let $k,d \in \mathbb{N}$ and suppose $d \geq k$.
 Let $G=(V,W,E)$ be a bipartite graph and $X = \{w_1,\dots,w_m\} \subseteq W$ be a $d$-closed set.
 Furthermore suppose that for distinct $v,v_1,\dots,v_k \in V$ we have~$|N(v) \cap (N(v_1)\cup N(v_2)\cup \dots \cup N(v_k))| \leq d-k$.
 Let $\bar x = (x_1,\dots,x_m)$ be a sequence of vertices with $x_i \in F(w_i)$ and let $\varphi \in \Aut(R(G)[[X]])$.
 Then $(R(G),\bar x) \simeq_k (R(G), \varphi(\bar x))$.
\end{lemma}

\begin{proof}
 We prove that Duplicator has a winning strategy in the bijective $(k+1)$-pebble game $\BP_{k+1}((R(G),\bar x),(R(G), \varphi(\bar x)))$ played on $(R(G),\bar x)$ and $(R(G), \varphi(\bar x))$.
 Towards this end we say a vertex $v \in V$ (respectively $w \in W$) is \emph{pebbled} if there exists $a \in M(v)$ (respectively $a \in F(w)$) which is pebbled.
 Furthermore we say that a vertex $w \in W$ is \emph{fixed} if there is some pebbled $v \in V$ with $\{v,w\} \in E(G)$.
 For a tuple $\bar a \in V(R(G))^{\leq k}$ of length at most~$k$ of pebbled vertices let
 \begin{align*}
  \cl(\bar a) = F(X) &\cup M(N^{-1}(X)) \cup \{M(v) \mid v \in V\colon v \text{ is pebbled}\}\\
                     &\cup \{F(w) \mid w \in W\colon w \text{ is pebbled or fixed}\}.
 \end{align*}
 Now let $\ell \leq k$, $\bar a, \bar b \in V(R(G))^{\ell}$ and suppose there is an isomorphism $\alpha\colon \cl(\bar a) \rightarrow \cl(\bar b)$ from $R(G)[\cl(\bar a)]$ to~$R(G)[\cl(\bar b)]$ mapping~$\bar x$ to~$\varphi(\bar x)$ and~$\bar  a$ to~$\bar  b$.
 Observe that $\alpha$ extends $\varphi$ and for $\ell = 0$ we can choose $\alpha = \varphi$.
 \begin{claim}
  For every unpebbled $v \in V$ with $N(v) \nsubseteq X$ there is some $w \in N(v) \setminus X$ which is neither pebbled nor fixed.
  \proof
  Consider $N(v) \setminus X$.
  Since $X$ is $d$-closed we conclude that $|N(v) \setminus X| \geq d+1$.
  By the assumption of the lemma there are at most~$d-k$ elements in~$N(v)$ that are fixed. 
  Thus, $N(v) \setminus X$ contains at least $k+1$ elements which are not fixed. Furthermore, there are at most~$k$ vertices in~$N(v)$ that are pebbled. Thus there is at least one element that is neither pebbled nor fixed.
  \uend
 \end{claim}
 For each unpebbled $v \in V$ with $N(v) \nsubseteq X$ choose a $w_v \in N(v) \setminus X$ that is neither pebbled nor fixed. Furthermore let $T = \{w \in X \mid \alpha(a(w)) = b(w)\}$. 
 For every~$a\in M(V)$ we define 
 \[B_a \coloneqq  \begin{cases}
         A \mathbin{\triangle} (T \cap N(v)) &\text{if } |T \cap N(v)| \text{ even, and}\\
         A \mathbin{\triangle} (T \cap N(v)) \mathbin{\triangle} \{w_v\} &\text{otherwise,}
        \end{cases}\]
 where~$A\subseteq N(a)$ is the set with~$m_A(v)= a$ and~$\mathbin{\triangle}$ denotes the symmetric difference. 
 We can now define the bijection
 \[f\colon V(R(G)) \rightarrow V(R(G))\colon a \mapsto \begin{cases}
                                                        \alpha(a) &\text{if } a \in \cl(\bar a)\\
                                                        a &\text{if } a \in F(W) \setminus \cl(\bar a)\\
                                                        m_{B_a}(v) &\text{if } a \in M(V) \setminus \cl(\bar a).
                                                       \end{cases}\]
 
 It is easy to check that for every $a \in V(R(G))$ there is an isomorphism $\alpha\colon \cl(\bar a,a) \rightarrow \cl(\bar b,f(a))$ 
 from $R(G)[\cl(\bar a,a)]$ to~$R(G)[\cl(\bar b,b)]$ mapping~$\bar x$ to~$\varphi(\bar x)$ and~$(\bar a,a)$ to~$(\bar b,b)$.
\end{proof}

\section{Meager graphs}

Searching for graphs in which we have control over the size of the~$d$-closure of a set we generalize the notion of an $\ell$-meager graph~\cite{DBLP:journals/jsyml/GurevichS96}.

\begin{definition}
 Let $G = (V,W,E)$ be a bipartite graph 
 and let $0 < \alpha < 1$.
 The graph $G$ is \emph{$(\ell,\alpha)$-meager} if for every $\emptyset \neq X \subseteq W$ with $|X| \leq \ell$ it holds that~$|N^{-1}(X)| < \alpha|X|$.
\end{definition}

Meager graphs have two properties that are advantageous for our course.
The first property is that for sufficiently small~$X \subseteq W$ the graph~$\Aut(R(G)[[X]])$ has many automorphisms.

\begin{lemma}
 \label{la:restricted-automorphism-size}
 Let $G=(V,W,E)$ be $(\ell,\alpha)$-meager and $X \subseteq W$ with $|X| \leq \ell$.
 Then \[|\Aut(R(G)[[X]])| \geq 2^{(1-\alpha)|X|}\;.\]
\end{lemma}

\begin{proof}
 By Lemma \ref{la:size-automorphism-group} for $X \subseteq W$ with $|X| \leq \ell$ we have that \[|\Aut(R(G)[[X]])| \geq 2^{|X| - |N^{-1}(X)|}\geq 2^{(1-\alpha)|X|}\;.\qedhere\]
\end{proof}

The second property that is advantageous for us is that in a meager graph the size of the~$d$-closure of a set~$X$ is only by a constant factor larger than~$|X|$ itself.

\begin{lemma}
 \label{la:size-closure}
 Suppose $d \in \mathbb{N}$ and $d\alpha < 1$.
 Let $G = (V,W,E)$ be $(\ell,\alpha)$-meager and suppose $\emptyset \neq X \subseteq W$ with $|X| \leq \ell(1 - d\alpha) -d +1$.
 Then $|\cl_G^{d}(X)| < \frac{1}{1-d\alpha}|X|$.
\end{lemma}

\begin{proof}
 Let $X_0\subsetneq \dots \subsetneq X_m$ be a sequence of sets such that $X_0 = X$ and $X_{i+1} = X_i \cup N(v_i)$ for some $v_i \in V$ with $|N(v_i) \setminus X_i| \leq d$ and such that $X_m$ is $d$-closed.
 Clearly, for every $i \in [m]$ it holds that $|N^{-1}(X_i)| \geq i$.
 Suppose that $m \geq \frac{\alpha}{1-d\alpha}|X|$ and set $j = \lceil\frac{\alpha}{1-d\alpha}|X|\rceil$.
 Then $|X_j| \leq |X| + dj \leq \lfloor\ell(1 - d\alpha)\rfloor - d + 1 + d\lceil\frac{\alpha}{1-d\alpha}\ell(1 - d\alpha)\rceil \leq \lfloor\ell(1 - d\alpha)\rfloor - d + 1 + d \lceil\ell \alpha\rceil  =  \ell + \lfloor -\ell d\alpha\rfloor - d + 1 + \lceil\ell d\alpha\rceil + d - 1 = \ell$. Hence the meagerness is applicable to~$X_j$. We conclude
 $j\leq |N^{-1}(X_j)| < \alpha|X_j| \leq \alpha(|X| + dj)$ implying~$j < \lceil \frac{\alpha}{1-d\alpha}|X|\rceil$. But this contradicts the definition of~$j$.
 So $m < \frac{\alpha}{1-d\alpha}|X|$ and thus, $|\cl_G(X)| = |X_m| \leq |X| + dm < \frac{1}{1-d\alpha}|X|$.
\end{proof}

We now concern ourselves with the existence of meager graphs. 
However, we require several additional properties. Indeed, 
in the light of Lemma~\ref{la:equivalent-tuples-from-automorphisms} we also want certain neighborhoods to be almost disjoint.
Moreover we want the graph $R(G)$ to only have few automorphisms, which by Lemma~\ref{la:make-rigid} translates into the matrix~$A_G$ having large rank.

\begin{theorem}
 \label{thm:meager-existence}
 There exists $r_0 \in \mathbb{N}$ such that for every $r \in \mathbb{N}$ with $r \geq r_0$ and every sufficiently large $n \in \mathbb{N}$ 
 there is an $\left(\frac{n}{10r},\frac{3}{r}\right)$-meager graph $G=(V,W,E)$ with
 \begin{enumerate}
  \item $|V| = |W| = n$,
  \item $\deg(v) = r$ for all $v \in V$,
  \item \label{item:meager-existence-3} $|N(v_1) \cap N(v_2)| < 3$ for all distinct $v_1,v_2 \in V$ and
  \item \label{item:meager-existence-4} $\rk(A_G) \geq (1 - 2^{-r})n$. 
 \end{enumerate}
\end{theorem}

Our proof of the theorem, which covers the rest of this section, makes use of the fact that bipartite expander graphs are meager. 

\begin{definition}
 Let $G=(V,W,E)$ be a bipartite graph with $|V| \geq |W|$.
 We call $G$ a \emph{$(\gamma,\beta)$-expander} if for every $Y \subseteq V$ with $|Y| \leq \gamma|V|$ it holds that $|N(Y)| \geq \beta|Y|$.
\end{definition}

One method to obtain bipartite expanders is by considering the following random process. Let $10 \leq r \in \mathbb{N}$ be a fixed number.
Given vertex sets $V$ and $W$ with $r \leq |W|/4$ and $|V| = |W|$ we obtain a bipartite graph $G = (V,W,E)$ by choosing independently and uniformly at random, for every $v \in V$ a set of $r$ distinct neighbors in~$W$.
We refer to~\cite[Section 4]{DBLP:journals/fttcs/Vadhan12} and \cite[Chapter 5.3]{DBLP:books/cu/MotwaniR95} for more information on expanders, including variants of the following lemma.

\begin{lemma}
 \label{la:expander}
 For~$r$ sufficiently large, $\Pr(G \text{ is a } \text{$\left(\frac{1}{10 r},\frac{r}{2}\right)$\text{-expander}}) \geq \frac{8}{9}$.
\end{lemma}

\begin{proof}
 Let $n = |V| = |W|$. For $Y \subseteq V$ and $X \subseteq W$ let~$p_{Y,X}$ denote the probability that that $N(Y) \subseteq X$.
 Then
 \[p_{Y,X} \leq \left(\frac{|X|}{n}\right)^{r\cdot |Y|}.\]
 Furthermore let $\beta = \frac{r}{2}$ and $\gamma = \frac{1}{10 r}$.
 Let $p$ be the probability that $G$ is not a $(\gamma,\beta)$-expander.
 Then, using the inequality~$\binom{n}{k}\leq (ne/k)^k$, we get
 \begin{align*}
  p &\leq \sum_{\substack{Y \subseteq V\\|Y| \leq \gamma\cdot n}}\;\sum_{\substack{X \subseteq W\\|X| = \lfloor \beta|Y|\rfloor}} p_{Y,X}\\
    &\leq \sum_{s = 1}^{\lfloor \gamma\cdot n\rfloor} \sum_{\substack{Y \subseteq V\\|Y| =s}}\;\sum_{\substack{X \subseteq W\\|X| =  \lfloor\beta|Y|\rfloor}} \left(\frac{|X|}{n}\right)^{r\cdot |Y|}\\
    &\leq \sum_{s = 1}^{\lfloor \gamma\cdot n\rfloor} \binom{n}{s}\binom{n}{\lfloor \beta s \rfloor} \left(\frac{\beta s}{n}\right)^{r\cdot s}\\
    &\leq \sum_{s = 1}^{\lfloor \gamma\cdot n\rfloor} \left(\frac{ne}{s}\right)^{s} \left(\frac{ne}{\beta s}\right)^{\beta \cdot s} \left(\frac{\beta s}{n}\right)^{r\cdot s}\\
    &= \sum_{s = 1}^{\lfloor \gamma\cdot n\rfloor} \left[\left(\frac{ne}{s}\right) \left(\frac{ne}{\beta s}\right)^{\beta} \left(\frac{\beta s}{n}\right)^{r}\right]^{s}
    \\   
    &= \sum_{s = 1}^{\lfloor \gamma\cdot n\rfloor} \left[\left(\frac{s}{n}\right)^{r-\beta-1} e^{1+\beta} \beta^{r-\beta}\right]^{s} \\
    &= \sum_{s = 1}^{\lfloor \gamma\cdot n\rfloor} \left[\left(\frac{s}{n}\right)^{r/2-1} e^{1+r/2} (r/2)^{r/2}\right]^{s}\\
    &\leq \sum_{s = 1}^{\lfloor \gamma\cdot n\rfloor} \left[\gamma^{r/2-1} e^{1+r/2} (r/2)^{r/2}\right]^{s}\;.
 \end{align*}
 Now let $x = \gamma^{r/2-1} e^{1+r/2} (r/2)^{r/2}$.
 If~$r$ is sufficiently large then $x = (10r)^{1-r/2} e^{1+r/2} (r/2)^{r/2}= 10er (e/20)^{r/2} <1/10$ . 
 It follows that \[p \leq \sum_{s = 1}^{\infty} x^{s} = \frac{x}{1-x} \leq \frac{1}{9}. \qedhere\]
\end{proof}

\begin{theorem}[cf.\ {\cite[Theorem 1.1]{calkin97}}]
 \label{thm:rank-expander}
 For $n \geq k$ let $S_{n,k} = \{v \in \mathbb{F}_2^{n} \mid |\{i \in [n] \mid v_i = 1\}| = k\}$.
 Furthermore let $A \in \mathbb{F}_2^{n \times n}$ be a random matrix where the rows are drawn uniformly and independently from $S_{n,k}$.
 There is a $K \in \mathbb{N}$ such that for every fixed $k \geq K$ it holds that
 \begin{equation*}
  \lim_{n \rightarrow \infty} \Pr(\rk(A) \geq (1 - 2^{-k})n) = 1.
 \end{equation*}

\end{theorem}

\begin{lemma}
 \label{la:hyperedge-intersection}
 Asymptotically almost surely there are no distinct vertices $v_1,v_2 \in V$ such that $|N(v_1) \cap N(v_2)| \geq 3$.
\end{lemma}

\begin{proof}
 Let $p$ be the probability that there are distinct $v_1,v_2 \in V$ with $|N(v_1) \cap N(v_2)| \geq 3$.
 Furthermore pick $c_3 > 0$ such that $\binom{n}{r-3} \leq c_3 \cdot n^{-3} \cdot \binom{n}{r}$ for all~$n\in \mathbb{N}$.
 Then, for~$n=|V| = |W|$ sufficiently large,
 \begin{align*}
  p &\leq |V|^{2} \sum_{s=0}^{r-3} \frac{\binom{|W|-r}{s}\binom{r}{r-s}}{\binom{|W|}{r}}
     \leq |V|^{2} \sum_{s=0}^{r-3} \frac{\binom{|W|}{r-3}\binom{r}{r-s}}{\binom{|W|}{r}}\\
    &\leq |V|^{2} (r-2) \binom{r}{\lfloor r/2 \rfloor} c_3 \cdot |W|^{-3}\\
    &\leq \frac{c_4}{|W|}
 \end{align*}
 for some constant $c_4 > 0$.
\end{proof}

\begin{proof}[Proof of Theorem \ref{thm:meager-existence}]
 Let $r$ be sufficiently large and let $G = (V,W,E)$ be a random bipartite graph as described above with $|W|/4 \geq r$ and $n= |V| = |W|$.
 By Lemma \ref{la:hyperedge-intersection}, Condition \ref{item:meager-existence-3} is satisfied asymptotically almost surely and Theorem \ref{thm:rank-expander} implies that Condition \ref{item:meager-existence-4} is satisfied asymptotically almost surely.
 So it remains to show $G$ is $\left(\frac{n}{10r},\frac{3}{r}\right)$-meager with a positive probability.
 By Lemma \ref{la:expander}, with probability at least~$8/9$, the graph $G$ is a $\left(\frac{1}{10 r},\frac{r}{2}\right)$-expander.
 Suppose there was an $\emptyset \neq X \subseteq W$ with $|X| \leq \frac{n}{10 r}$ for which $N^{-1}(X) \geq \frac{3}{r} \cdot |X|$. 
 Then we could choose $Y \subseteq N^{-1}(X)$ with~$|Y| = \frac{3}{r} \cdot |X|\leq \frac{n}{10 r}$ (assuming~$r\geq 3$). But then  $|X| \geq |N(Y)| \geq \frac{r}{2}\cdot |Y| \geq \frac{r}{2} \cdot \frac{3}{r} \cdot |X| > |X|$ which is a contradiction. 
\end{proof}

\section{Lower bounds for individualization-refinement algorithms}

In the previous section we have proven the existence of meager graphs with various additional properties. We show now that applying the multipede construction to such graphs yields examples where the search tree of individualization-refinement algorithms is large.

%

\begin{lemma}
 \label{la:number-equivalent-tuples}
 Let $k,d \in \mathbb{N}$ and suppose $d \geq k$ and $d\alpha < 1$.
 Let $G = (V,W,E)$ be $(\ell,\alpha)$-meager and let $X = \{w_1,\dots,w_m\} \subseteq W$ be a subset of cardinality $m \leq (1-d\alpha)\ell - d + 1$.
 Furthermore suppose that for all distinct $v,v_1,\dots,v_k \in V$ we have~$|N(v) \cap (N(v_1)\cup N(v_2)\cup \dots \cup N(v_k))| \leq d-k $.
 Let $\bar x = (x_1,\dots,x_m)$ be a sequence with $x_i \in F(w_i)$. Then
 \[|\{\bar y \in  F(W)^{m} \mid (R(G),\bar x) \simeq_k (R(G),\bar y)\}| \geq 2^{\frac{1-\alpha(d+1)}{1-d\alpha}m}\;.\]
\end{lemma}

\begin{proof}
 Let $\widehat{X} = \cl_G^{d}(X)$ be the $d$-closure of $X$.
 By Lemma \ref{la:size-closure} we get that $|\widehat{X}| < \frac{1}{1-d\alpha}|X| \leq \ell$.
 Suppose $\widehat{X} = X \cup \{u_1,\dots,u_s\}$ and let $\bar z = (x_1,\dots,x_m,z_1,\dots,z_{s})$ be an extension of $\bar x$ with $z_i \in F(u_i)$.
 By Lemmas \ref{la:equivalent-tuples-from-automorphisms} and \ref{la:restricted-automorphism-size} we conclude that
 \[|\{\bar y \in  F(W)^{m+s} \mid (R(G),\bar z) \simeq_k (R(G),\bar y)\}| \geq |\Aut(R(G)[[\widehat{X}]])| \geq 2^{(1-\alpha)(m+s)}\;.\]
 Let $A = \{\bar y \in  F(W)^{m+s} \mid (R(G),\bar z) \simeq_k (R(G),\bar y)\}$ and for $\bar a \in A$ let $\pi_{m}(\bar a)$ be the projection onto the first $m$ components. 
 Clearly, for $\bar a,\bar b \in A$, it holds that $(R(G),\pi_m(\bar a)) \simeq_k (R(G),\pi_m(\bar b))$.
 So
 \[|\{\bar y \in  F(W)^{m} \mid (R(G),\bar x) \simeq_k (R(G),\bar y)\}| \geq |A|\cdot2^{-s} \geq 2^{(1-\alpha)(m+s) - s} = 2^{(1-\alpha)m - \alpha s}\;.\]
 Since $s \leq m - \frac{1}{1-d\alpha}m =  \frac{d\alpha m}{1-d\alpha}$ we conclude that
 \[2^{(1-\alpha)m - \alpha s} \geq 2^{(1-\alpha)m - \frac{d\alpha^{2} m}{1-d\alpha}} = 2^{\frac{1-\alpha(d+1)}{1-d\alpha}m}\;.\]
\end{proof}

\begin{theorem}
 \label{thm:lower-bound-from-meager}
 Let $k \in \mathbb{N}$, $\ell \geq \max\{9r,45s, 18k\}$ and $\alpha\leq \frac{1}{10k}$.
 Suppose $G = (V,W,E)$ is a bipartite graph with $n = |V| = |W|$ such that
 \begin{enumerate}
  \item $G$ is $(\ell,\alpha)$-meager,
  \item $\deg(v) = r$ for all $v \in V$,
  \item $|N(v_1) \cap N(v_2)| < 3$ for all distinct $v_1,v_2 \in V$ and
  \item $n - \rk(A_G) \leq s$.
 \end{enumerate}
 Then there is a subset $I \subseteq W$ with $|I| \leq s$ such that
 \begin{enumerate}
  \item $R^{I}(G)$ is rigid and
  \item for every $k$-realizable cell selector $\sel$, every $k$-realizable node invariant $\inv$ and every $k$-realizable refinement operator $\refine$  it holds that
   \begin{equation}
    |\mathcal{T}_{\inv}^{\refine,\sel}(R^{I}(G))| \geq 2^{\frac{1}{36}\cdot \ell}.
   \end{equation}
 \end{enumerate}
\end{theorem}

\begin{proof}
 Set $d \coloneqq 3k$.
 By Lemma \ref{la:make-rigid} there is a set $I \subseteq W$ of size $|I| \leq s$ such that $R^{I}(G)$ is rigid.
 Suppose $I = \{w_1,\dots,w_s\}$.
 \begin{claim}
  For distinct $v,v_1,\dots,v_k \in V$ we have~$|N(v) \cap (N(v_1)\cup N(v_2)\cup \dots \cup N(v_k))| \leq d-k $.
  \proof
  We have~$N(v) \cap (N(v_1)\cup N(v_2)\cup \dots \cup N(v_k)) \leq 2k\leq d-k$. 
  \uend
 \end{claim}
 
 We choose an arbitrary linear order on the set of vertices in~$W$.
 For a vertex~$m_A(v)\in R(G)$ define the projection~$\Pi(m_A(v))$ to be the sequence~$(x_1,\ldots,x_{|N(m_A(v))|})$ of the vertices in~$N(m_A(v))$ ordered according to the linear order of~$W$ (observe that for each $w \in N(v)$ either $a(w)$ or $b(w)$ occurs in the sequence).
 We extend~$\Pi$ to~$V(R(G))$ by defining for~$w\in W$ that~$\Pi(a(w))=(a(w))$ and~$\Pi(b(w))=(b(w))$.
 
 For a sequence~$\bar{x}=(x_1,\dots,x_t)$ of vertices of~$R(G)$ we define~$\widehat{\Pi(\bar{x})}$ as the concatenation of the sequences~$\Pi(x_1),\ldots,\Pi(x_t)$.
 We let~${\Pi(\bar{x})}$ be the subsequence of~$\widehat{\Pi(\bar{x})}$ in which all duplicates are removed starting from right. (I.e, for a sequence $\bar y = (y_1,\dots,y_t)$ we inductively the sequence $\rho(\bar y)$ by setting $\rho(\varepsilon) = \varepsilon$ and letting~${\rho(y_1,\ldots,y_t)}$ be equal to~${\rho(y_1,\ldots,y_{t-1})}$, if~$y_i=y_t$ for some~$i<t$, and is equal to the concatenation of~${\rho(y_1,\ldots,y_{t-1})}$ and~$y_t$, otherwise. Then $\Pi(\bar x) = \rho(\widehat{\Pi(\bar{x})})$.)
 
 \begin{claim}[resume] For every sequence~$\bar{x}$ of vertices of~$R(G)$ it holds that
  \begin{align*}
       |\{\bar y \mid (R^{I}(G),\bar y) \simeq_k (R^{I}(G),\bar x)\}| 
       & \geq   |\{\bar z \mid (R^{I}(G),\bar z) \simeq_k (R^{I}(G),\Pi(\bar{x}))\}| \\ 
         &\geq  2^{-s}\cdot |\{\bar z \mid (R(G),\bar z) \simeq_k (R(G),\Pi(\bar{x}))\}|.
  \end{align*}
  \proof
  Observe first that the second inequality holds
  since there are only~$2^s$ color preserving permutations of~$F(I)$.
  
  For the first inequality suppose for~$\bar{z}$ we have~$(R^{I}(G),\bar z) \simeq_k (R^{I}(G),\Pi(\bar x))$ then we can find a lift~$\bar{y} = (y_1,\ldots,y_i)$ with~$\Pi(\bar{y}) = \bar{z}$ so that~$y_j$ and~$x_j$ have the same color for all~$j\in \{1,\ldots,i\}$. For this lift we have~$(R^{I}(G),\bar y) \simeq_k (R^{I}(G),\bar x)$. Since lifts of distinct sequences must be distinct we conclude the claim.
  \uend
 \end{claim}
 
 Now let $t = \lfloor(1-d\alpha)\ell - d + 1\rfloor \geq \lfloor(1-\frac{3k}{10k} )\ell - d + 1\rfloor \geq   \lfloor \frac{1}{2}\ell \rfloor - d + 1 \geq \frac{1}{3}\ell$ and let~$c \coloneqq \frac{1-\alpha(d+1) }{1-d\alpha}  = \frac{1-\alpha(3k+1) }{1-3k\alpha}  \geq \frac{1}{2}$.
 
 By Lemma \ref{la:number-equivalent-tuples} and Claim 2 we conclude that for every vertex sequence~$\bar{x}$ with~$|\Pi(\bar{x})|\leq t-s$ we have
 \begin{equation}
  \label{eq:many-equivalent-tuples}
  |\{\bar z \mid (R^{I}(G),\bar z) \simeq_k (R^{I}(G),\Pi(\bar x))\}| \geq 2^{c{|\bar{x}|/2}-s} \geq 2^{{|\bar{x}|/4}-s}.
 \end{equation}
 (Here the extra~$1/2$ in the exponent comes from the fact that in Lemma \ref{la:number-equivalent-tuples} for each~$w_i$ only one $x_i \in F(w_i)$ can be chosen but here~$\Pi(\bar{x})$ can contain both vertices from~$ F(w_i)$.)
 \begin{claim}[resume] There is a sequence~$\bar{y}$ in~$\mathcal{T}_{\inv}^{\refine,\sel}(R^{I}(G))$ with~$ t-s-r < |\Pi(\bar y)| \leq t-s$.
  \proof
  Let $\bar x$ be a leaf of $\mathcal{T}_{\inv}^{\refine,\sel}(R^{I}(G))$ and let $m = |\bar x|$ be the length of $\bar x$. Observe that $\pi_i(\bar x) \coloneqq (x_1,\dots,x_i) \in V(\mathcal{T}_{\inv}^{\WL_k,\sel}(R^{I}(G)))$ for every $i \in [m]$.  Define~$t_i$ as~$|\Pi(\pi_i(\bar x))|$.
  Note that~$t_i\leq t_{i-1} + r$ since~$|\Pi(x_i)| \leq r$. It thus suffices to show that~$|\Pi(\bar x)| > t-s-r$.
  Assume otherwise. We show that~$\refine(R^{I}(G),\bar{x})$ is not discrete and thus~$\bar{x}$ is not a leaf.
  Let~$\bar{x}'$  be a sequence of which~$\bar{x}$ is a prefix that satisfies~$4s < |\Pi(\bar{x}')| \leq t-s$ (possibly~$\bar{x}'= \bar{x}$). It suffices now to show that~$\refine(R^{I}(G),\bar{x}')$ is not discrete. 
  Indeed, by Equation (\ref{eq:many-equivalent-tuples}) we have~$|\{\bar z \mid (R^{I}(G),\bar z) \simeq_k (R^{I}(G),\Pi(\bar x'))\}| \geq 2^{1/4|\Pi(\bar x')|-s} > 2^{s -s} =1$. 
  Since~$R^{I}$ is rigid, this implies that~$\refine(R^{I}(G),\bar{x}')$ is not discrete.  
  \uend
 \end{claim}
 Applying the sequence~$\bar{y}$ from Claim 3 to Equation (\ref{eq:many-equivalent-tuples}) we obtain that~$|\{\bar z \mid (R^{I}(G),\bar z) \simeq_k (R^{I}(G),\Pi(\bar y))\}| \geq 2^{c{|\bar{y}|/2}-s} \geq 2^{{|\bar{y}|/4}-s}\geq 2^{t/4-r/4-5/4s} \geq 2^{\ell/12-\ell/36-\ell/36}\geq 2^{\ell/36} $. By Claim 2 this means that~$|\{\bar z \mid (R^{I}(G),\bar z) \simeq_k (R^{I}(G), \bar y)\}|\geq 2^{\ell/36} $. We conclude the proof by an application of Lemma~\ref{la:tree-nodes-from-equivalent-tuples}.
\end{proof}

Having already shown  in the previous section the existence of graphs that satisfy the requirements of the theorem we just proved, we can now prove the main theorem. 

\begin{proof}[Proof of Theorem~\ref{thm:main}]
Theorem~\ref{thm:main} now follows by combining Theorems \ref{thm:meager-existence} and \ref{thm:lower-bound-from-meager}.
\end{proof}

While the graphs we have constructed are colored graphs, we should remark that it is not difficult to turn them into uncolored graphs while preserving the exponential size of the search tree. 
Indeed, we form the disjoint union of the graph~$R^{I}(G)$ with a path of length~$t+1$ where~$t$ is the number of colors in~$R^{I}(G)$.
We then order the colors and connect the~$i$-th vertex of the path with all vertices of color~$i$.
Finally we add a vertex adjacent to all but the last vertex of the path to obtain a graph~$\widetilde{R^{I}(G)}$.
In the resulting uncolored graph~$\widetilde{R^{I}(G)}$, the last vertex of the path is the only vertex with degree~1.
Moreover if we apply color-refinement then all newly added vertices are singletons and the partition induced by color classes on~$V(R^{I}(G))$ in~$\widetilde{R^{I}(G)}$ is the same as it the one in~$R^{I}(G)$.
This shows that the graph is still rigid.
It also implies that each search tree of~$\widetilde{R^{I}(G)}$ corresponds to a search tree of~$R^{I}(G)$ of the same size.

\section{Component Recursion}

Component recursion is a mechanism that can be used as addition to an individualization-refinement algorithm improving their performance.
With the right cell selection strategy, using component recursion,  individualization-refinement algorithms have exponential upper bounds~\cite{Goldberg1983229}.

Some form of component recursion is for example implemented in the newest version of bliss and was demonstrated to yield significant improvements in practice~\cite{DBLP:conf/tapas/JunttilaK11}.
In this section we argue that for our examples it is not possible to beat the exponential lower bounds even with the use of component recursion, which we explain first. 

For a graph $G=(V,E)$ we say that two disjoint vertex sets $X,Y \subseteq V$ are \emph{uniformly joined} if~$E\cap (X\times Y) = \emptyset$ or~$E\cap (X\times Y) = X\times Y$. 
By definition the sets are uniformly joined if $X = \emptyset$ or $Y = \emptyset$.

\begin{definition}
 Let $G = (V,E,c)$ be a colored graph and let $S \subseteq V$ be a set of vertices.
 We say that $S$ is a \emph{color-component of $G$} if for all colors $i,j \in [n]$ the sets~$S \cap c^{-1}(i)$ and~$c^{-1}(j) \setminus S$ are uniformly joined.
\end{definition}

We will not go into great detail on how color-components can be exploited by individualization-refinement algorithms and rather refer to~\cite{DBLP:conf/tapas/JunttilaK11}.
Intuitively, components can be used to treat parts of the input graph independently. 
For us it will be sufficient to note the following.
If a color-component is a union of color classes (of the stable coloring under 1-dimensional Weisfeiler-Leman) then a suitable cell selector can ensure that the component is explored entirely before the search progresses into vertices outside of the color-component~(see~\cite{DBLP:conf/tapas/JunttilaK11}).
We will show now that for the graphs that appear as nodes in the search tree of our graphs there are no color-components besides those that are unions of color classes.

\begin{definition}
 Let $G = (V,E,c)$ be a colored graph and let $X \subseteq V$ be a set of vertices.
 We say that $X$ \emph{respects color $i \in [n]$} if either $c^{-1}(i) \cap X = \emptyset$ or $c^{-1}(i) \subseteq X$.
\end{definition}

Let $G = (V,E,c)$ be a colored graph and let $\chi \colon V \rightarrow [n]$ be a second coloring of the vertices.
We call $\chi$ an \emph{equitable} coloring of $G$ if $\chi \preceq c$ and for all $i,j \in [n]$ and all $v,w \in \chi^{-1}(i)$ it holds that \[|N(v) \cap \chi^{-1}(j)| = |N(w) \cap \chi^{-1}(j)|.\]
Observe that a coloring $\chi$ is equitable if and only if it is not further refined by $1$-dimensional Weisfeiler-Leman, that is, $\chi$ is stable with respect to the $1$-dimensional Weisfeiler-Leman algorithm.

\begin{lemma}
 Let $G =(V,W,E)$ be a bipartite graph and let $\chi$ be an equitable coloring of $R(G)$.
 Suppose $S \subseteq V(R(G))$ is a color-component of $R^{\chi}(G) = (V(R(G)),E(R(G)),\chi)$.
 Then $S$ is a union of color classes of~$\chi$ or the graph $R^{\chi}(G)$ has a non-trivial automorphism.
\end{lemma}

\begin{proof}
 Suppose that $S$ is not a union of color classes.
 Then there is some color $i \in [n]$ such that $C_1 := S \cap \chi^{-1}(i) \neq \emptyset$ and $C_2 := \chi^{-1}(i) \setminus S \neq \emptyset$.
 Let $c$ be the (original) coloring of $R(G)$.
 Since $\chi \preceq c$, there either is a vertex $w \in W$ with $C_1 \cup C_2 \subseteq F(w)$ or there is some $v \in V$ with $C_1 \cup C_2 \subseteq M(v)$.
 
 \begin{claim}
  $S$ respects all colors $i \in [n]$ for which~$|\chi^{-1}(i)| \geq 3$.
  \proof
  Suppose otherwise that~$S$ does not respect~$i$ and choose $v \in V$ such that $\chi^{-1}(i) \subseteq M(v)$.
  Furthermore let $A_1,A_2,A_3 \subseteq N(v)$ be distinct subsets of even size such that $\chi(m_{A_j}(v)) = i$ for all $j \in \{1,2,3\}$, $m_{A_1}(v) \in S$ and $m_{A_2}(v) \notin S$.
  Choose $w_1 \in A_1 \setminus A_2$ and $w_2 \in A_2 \setminus A_1$.
  First observe that $F(w_1)$ and $F(w_2)$ form color classes with respect to $\chi$, because $\chi$ is equitable.
  Furthermore, $S$ does not respect the color classes $F(w_1)$ and $F(w_2)$.
  
  Now assume without loss of generality that $w_1 \in A_3$. We distinguish two cases.
  First suppose that $w_2 \in A_3$.
  If $m_{A_3}(v) \in S$ then $S$ has to respect the color the color class $F(w_2)$.
  Otherwise $m_{A_3}(v) \notin S$ and $S$ has to respect the color the color class $F(w_1)$.
  In both cases we obtain a contradiction.
  
  So $w_2 \notin A_3$. Let $w_3 \in A_1 \setminus A_3$.
  If $m_{A_3}(v) \in S$ then $S$ has to respect the color the color class $F(w_3)$.
  But $m_{A_2}(v)$ is only connected to one of the vertices $\{a(w_3),b(w_3)\}$ which again leads to a contradiction.
  So $m_{A_3}(v) \notin S$.
  But then, looking at $m_{A_2}(v)$ and $m_{A_3}(v)$, $S$ has to respect the color class $F(w_2)$.
  Once again, this is a contradiction.
  \uend
 \end{claim}
 
 Now let $M = \{i \in [n]\mid \text{$S$ does not respect color $i$}\}$.
 For each $i \in M$ we have $\chi^{-1}(i) = \{a_i,b_i\}$ where $a_i \in S$ and $b_i \in \overline{S}$.
 We define \[\gamma \colon V(R(G)) \rightarrow V(R(G))\colon v \rightarrow \begin{cases}
                                                                            v & \text{if } \chi(v) \notin M\\
                                                                            b_i & \text{if } v = a_i \text{ for some } i \in M\\
                                                                            a_i & \text{if } v = b_i \text{ for some } i \in M
                                                                           \end{cases}.\]
 \begin{claim}[resume]
  $\gamma \in \Aut(R^{\chi}(G))$.
  \proof
  Let $\{u,m\} \in E(R(G))$ where $u \in F(w)$ for some $w \in W$ and $m \in M(v)$ for some $v \in N(w)$.
  First suppose $\gamma(u) \neq u$.
  Then $\chi(u) \in M$ and $S$ does not respect the color class $F(w)$. Without loss of generality assume $u = a(w)$ and thus, $\gamma(u) = b(w)$.
  Since $\chi$ is equitable there is some $m' \in N_{R(G)}(b(w)) \setminus N_{R(G)}(a(w))$ with $\chi(m) = \chi(m')$.
  But then $|S \cap \{m,m'\}| = 1$ because $S$ is a color-component.
  Hence, $S$ does not respect color $\chi(m)$ and $\chi^{-1}(\chi(m)) = \{m,m'\}$ by Claim 1.
  So $\gamma(m) = m'$ and thus, $\{\gamma(u),\gamma(m)\} \in E(R(G))$.
  
  Similarly, it also follows that $\{\gamma(u),\gamma(m)\} \in E(R(G))$ if $\gamma(m) \neq m$. So $\gamma$ preserves the edge relation of $R(G)$ and thus, $\gamma \in \Aut(R^{\chi}(G))$.
  \uend
 \end{claim}
\end{proof}

\begin{corollary}
Let $G =(V,W,E)$ be a bipartite graph. If $R^{I}(G)$ is rigid then every color-component of $R^{I}(G)$ of an equitable coloring is a union of color classes.
\end{corollary}

As we mentioned before, color-components that are unions of color classes can be handled with a suitable cell selector.
We conclude that the lower bound of Theorem~\ref{thm:main} also applies to individualization-refinement algorithms that apply component recursion.

\section{Discussion}

We presented a construction resulting in graphs for which the search tree of individualization-refinement algorithms has exponential size.
As a consequence algorithms based on this paradigm have exponential worst case complexity.
In particular, this includes the to date fastest practical isomorphism solvers such as nauty, traces, bliss, saucy and conauto.

While the analysis in this paper is theoretical, constructions related to the ones presented in this paper produce graphs that constitute the practically most difficult instances to date~\cite{benchmark-paper}.

Our construction gives, for each constant $k$, a family of graphs that lead to exponential size search trees when the $k$-dimensional Weisfeiler-Leman algorithm is used as a refinement operator.
However, it is not clear whether we can obtain similar statements if the Weisfeiler-Leman dimension may depend on the number of vertices of the input graph.
The recent quasi-polynomial time algorithm due to Babai \cite{DBLP:conf/stoc/Babai16} for example repeatedly benefits from performing the $\Theta(\log n)$-dimensional Weisfeiler-Leman algorithm where $n$ is the number of vertices of the input graphs.
In particular, it is an interesting question whether we still obtain exponential size search trees if we allow the Weisfeiler-Leman dimension to be linear in the number of vertices.
A related question asks for the maximum Weisfeiler-Leman dimension of rigid graphs.

Another interesting question concerns the complexity of individualization-refinement algorithms for other types of structures.
Here, we are particularly interested in the group isomorphism problem (for groups given by multiplication table). What is the running time of the individualization-refinement paradigm for groups?
In fact, it is not even known whether groups have bounded Weisfeiler-Leman dimension. That is, whether there is a~$k$ for which the~$k$-dimensional Weisfeiler-Leman algorithm solves group isomorphism.

Finally, all known constructions leading to graphs with large Weisfeiler-Leman dimension are in some way based on the CFI-construction, but it would be desirable to have constructions that are conceptually different.
Let us remark again that isomorphism for the graphs we constructed can be decided im polynomial time using techniques from algorithmic group theory.
It would be interesting to find explicit constructions that are difficult for both group theoretic techniques and methods based on the Weisfeiler-Leman algorithm.

\bibliographystyle{abbrv}
\bibliography{literature}

\begin{thebibliography}{10}

\bibitem{BabaiRandom}
L.~Babai.
\newblock Monte carlo algorithms in graph isomorphism testing.
\newblock Technical Report 79-{10}, Universit\'e de Montr\'eal, 1979.

\bibitem{DBLP:conf/stoc/Babai16}
L.~Babai.
\newblock Graph isomorphism in quasipolynomial time [extended abstract].
\newblock In D.~Wichs and Y.~Mansour, editors, {\em Proceedings of the 48th
  Annual {ACM} {SIGACT} Symposium on Theory of Computing, {STOC} 2016,
  Cambridge, MA, USA, June 18-21, 2016}, pages 684--697. {ACM}, 2016.

\bibitem{DBLP:conf/stoc/BabaiL83}
L.~Babai and E.~M. Luks.
\newblock Canonical labeling of graphs.
\newblock In D.~S. Johnson, R.~Fagin, M.~L. Fredman, D.~Harel, R.~M. Karp,
  N.~A. Lynch, C.~H. Papadimitriou, R.~L. Rivest, W.~L. Ruzzo, and J.~I.
  Seiferas, editors, {\em Proceedings of the 15th Annual {ACM} Symposium on
  Theory of Computing, 25-27 April, 1983, Boston, Massachusetts, {USA}}, pages
  171--183. {ACM}, 1983.

\bibitem{cfi}
J.~Cai, M.~F{\"{u}}rer, and N.~Immerman.
\newblock An optimal lower bound on the number of variables for graph
  identifications.
\newblock {\em Combinatorica}, 12(4):389--410, 1992.

\bibitem{calkin97}
N.~J. Calkin.
\newblock Dependent sets of constant weight binary vectors.
\newblock {\em Combinatorics, Probability {\&} Computing}, 6(3):263--271, 1997.

\bibitem{saucy}
P.~Codenotti, H.~Katebi, K.~A. Sakallah, and I.~L. Markov.
\newblock Conflict analysis and branching heuristics in the search for graph
  automorphisms.
\newblock In {\em 2013 {IEEE} 25th International Conference on Tools with
  Artificial Intelligence, Herndon, VA, USA, November 4-6, 2013}, pages
  907--914. {IEEE} Computer Society, 2013.

\bibitem{DBLP:conf/focs/FurstHL80}
M.~L. Furst, J.~E. Hopcroft, and E.~M. Luks.
\newblock Polynomial-time algorithms for permutation groups.
\newblock In {\em 21st Annual Symposium on Foundations of Computer Science,
  Syracuse, New York, USA, 13-15 October 1980}, pages 36--41. {IEEE} Computer
  Society, 1980.

\bibitem{Goldberg1983229}
M.~Goldberg.
\newblock A nonfactorial algorithm for testing isomorphism of two graphs.
\newblock {\em Discrete Applied Mathematics}, 6(3):229 -- 236, 1983.

\bibitem{DBLP:journals/jsyml/GurevichS96}
Y.~Gurevich and S.~Shelah.
\newblock On finite rigid structures.
\newblock {\em J. Symb. Log.}, 61(2):549--562, 1996.

\bibitem{IL90}
N.~Immerman and E.~Lander.
\newblock {\em Describing Graphs: A First-Order Approach to Graph
  Canonization}, pages 59--81.
\newblock Springer New York, New York, NY, 1990.

\bibitem{bliss}
T.~Junttila and P.~Kaski.
\newblock Engineering an efficient canonical labeling tool for large and sparse
  graphs.
\newblock In D.~Applegate, G.~S. Brodal, D.~Panario, and R.~Sedgewick, editors,
  {\em Proceedings of the Ninth Workshop on Algorithm Engineering and
  Experiments and the Fourth Workshop on Analytic Algorithms and
  Combinatorics}, pages 135--149. SIAM, 2007.

\bibitem{DBLP:conf/tapas/JunttilaK11}
T.~A. Junttila and P.~Kaski.
\newblock Conflict propagation and component recursion for canonical labeling.
\newblock In A.~Marchetti{-}Spaccamela and M.~Segal, editors, {\em Theory and
  Practice of Algorithms in (Computer) Systems - First International {ICST}
  Conference, {TAPAS} 2011, Rome, Italy, April 18-20, 2011. Proceedings},
  volume 6595 of {\em Lecture Notes in Computer Science}, pages 151--162.
  Springer, 2011.

\bibitem{conauto}
J.~L. L{\'{o}}pez{-}Presa, A.~F. Anta, and L.~N. Chiroque.
\newblock Conauto-2.0: Fast isomorphism testing and automorphism group
  computation.
\newblock {\em CoRR}, abs/1108.1060, 2011.

\bibitem{MR635936}
B.~D. McKay.
\newblock Practical graph isomorphism.
\newblock {\em Congr. Numer.}, 30:45--87, 1981.

\bibitem{mckay}
B.~D. McKay and A.~Piperno.
\newblock Practical graph isomorphism, {II}.
\newblock {\em J. Symb. Comput.}, 60:94--112, 2014.

\bibitem{DBLP:conf/dimacs/Miyazaki95}
T.~Miyazaki.
\newblock The complexity of {McKay's} canonical labeling algorithm.
\newblock In {\em Groups and Computation}, volume~28 of {\em {DIMACS} Series in
  Discrete Mathematics and Theoretical Computer Science}, pages 239--256.
  {DIMACS/AMS}, 1995.

\bibitem{DBLP:books/cu/MotwaniR95}
R.~Motwani and P.~Raghavan.
\newblock {\em Randomized Algorithms}.
\newblock Cambridge University Press, 1995.

\bibitem{benchmark-paper}
D.~Neuen and P.~Schweitzer.
\newblock Benchmark graphs for practical graph isomorphism.
\newblock 2017.
\newblock manuscript.

\bibitem{DBLP:journals/fttcs/Vadhan12}
S.~P. Vadhan.
\newblock Pseudorandomness.
\newblock {\em Foundations and Trends in Theoretical Computer Science},
  7(1-3):1--336, 2012.

\end{thebibliography}

\end{document}